\documentclass[twoside,leqno,twocolumn]{article}

\usepackage[letterpaper]{geometry}

\usepackage{ltexpprt}

\usepackage{xspace}
\usepackage{algorithmicx}
\usepackage{subcaption,graphicx,url,multirow,multicol,mathtools}
\usepackage[noend]{algpseudocode}
\usepackage{balance, afterpage}
\usepackage{paralist}
\usepackage{booktabs}
\usepackage{color}
\usepackage{arydshln}
\usepackage{enumitem}
\usepackage[most]{tcolorbox}
\usepackage{stmaryrd}

\usepackage{bm}

\newcommand{\myparagraph}[1]{\vspace{1.5pt}\noindent {\bf #1.}}

\mathchardef\mhyphen="2D

\newcommand{\whp}[1]{\textit{whp}}
\newcommand{\defn}[1]{{\bf{\emph{#1}}}}

\usepackage[ruled]{algorithm}

\newcommand{\BigO}[1]{\ensuremath{O(#1)}}
\newcommand{\BigOmega}[1]{\ensuremath{\Omega(#1)}}
\newcommand{\pcomplete}{\ensuremath{\mathsf{P}}-complete}
\newcommand{\nc}{\ensuremath{\mathsf{NC}}}

\newcommand{\algname}[1]{\textnormal{\textsc{#1}}}
\newcommand{\algprefix}{Algorithm\xspace}
\newcommand{\kclist}{\algname{kClist}\xspace}
\newcommand{\pivoter}{\algname{Pivoter}\xspace}
\newcommand{\wco}{\algname{WCO}\xspace}
\newcommand{\binaryjoin}{\algname{BinaryJoin}\xspace}
\newcommand{\coreapp}{\algname{CoreApp}\xspace}

\newcommand{\ourcount}{\algname{arb-count}\xspace}
\newcommand{\ourpeel}{\algname{arb-peel}\xspace}
\newcommand{\ourapproxpeel}{\algname{arb-approx-peel}\xspace}
\newcommand{\ourcountv}{\algname{arb-count-v}\xspace}

\newcommand{\intheapp}{appendix}

\newcommand{\kc}{$k$-clique\xspace}
\newcommand{\kcds}{$k$-clique densest subgraph\xspace}

\algblock{ParFor}{EndParFor}
\algnewcommand\algorithmicparfor{\textbf{parfor}}
\algnewcommand\algorithmicpardo{\textbf{do}}
\algnewcommand\algorithmicendparfor{}
\algrenewtext{ParFor}[1]{\algorithmicparfor\ #1\ \algorithmicpardo}
\algrenewtext{EndParFor}{\algorithmicendparfor}
\algtext*{EndParFor}

\algdef{SE}[DOWHILE]{Do}{doWhile}{\algorithmicdo}[1]{\algorithmicwhile\ #1}%

\definecolor{ao}{rgb}{0.0, 0.5, 0.0}
\newcommand{\tblsty}[1]{{\bf{\textcolor{ao}{#1}}}}

\usepackage[font=small,aboveskip=0pt,belowskip=0pt]{caption}

\usepackage{times,microtype}

\usepackage{thmtools}
\usepackage{thm-restate}
\usepackage{amsfonts}

\begin{document}

\title{Parallel Clique Counting and Peeling Algorithms} 

\author{Jessica Shi\thanks{MIT CSAIL, Cambridge, MA (jeshi@mit.edu, laxman@mit.edu, jshun@mit.edu)}
  \and Laxman Dhulipala\footnotemark[1]
\and Julian Shun\footnotemark[1]}

\date{}

\maketitle

\fancyfoot[R]{\scriptsize{Copyright \textcopyright\ 2021 by SIAM\\
Unauthorized reproduction of this article is prohibited}}

\begin{abstract}

We present a new parallel algorithm for $k$-clique counting/listing
that has polylogarithmic span (parallel time) and is work-efficient (matches the work of the best sequential algorithm)
for sparse graphs.
Our algorithm is based on computing low out-degree orientations, which we present 
new linear-work and polylogarithmic-span algorithms for computing in parallel.
We also present new parallel algorithms for
producing unbiased estimations of clique counts using graph
sparsification. Finally, we design two new parallel work-efficient
algorithms for approximating the $k$-clique densest subgraph,
the first of which is a $1/k$-approximation
and the 
second of which is a $1/(k(1+\epsilon))$-approximation and has polylogarithmic span. Our first algorithm does not have polylogarithmic span, but we prove that it solves
 a \pcomplete{} 
 problem.

In addition to the theoretical results, we also implement the algorithms and
propose various optimizations to improve their practical performance.
On a 30-core machine with two-way hyper-threading, our algorithms
achieve 13.23--38.99x and 1.19--13.76x self-relative parallel speedup
for $k$-clique counting and $k$-clique densest subgraph,
respectively. Compared to the state-of-the-art parallel $k$-clique
counting algorithms, we achieve up to 9.88x speedup, and compared to
existing implementations of $k$-clique densest subgraph, we achieve up 
to 11.83x speedup.  We are able to compute the $4$-clique counts on
the largest publicly-available graph with over two hundred billion
edges for the first time.

\end{abstract}

\section{Introduction}

Finding \kc{}s in a graph is a fundamental
graph-theoretic problem with a long history of study both in theory
and practice. In recent years, \kc counting and listing have
been widely applied in practice due to their many applications,
including in learning network embeddings~\cite{rossi2018hone},
understanding the structure and formation of
networks~\cite{yin2018higher, TsPaMi17}, identifying dense subgraphs for
community detection~\cite{Tsourakakis15, Sariyuce2017, FaYuChLaLi19,Gregori2013}, and
graph partitioning and compression~\cite{Feder95}.

For sparse graphs, the best known sequential algorithm is by Chiba and
Nishizeki~\cite{ChNi85}, and requires $O(m\alpha^{k-2})$ work
(number of operations), where $\alpha$ is the arboricity of the
graph.\footnote{A graph has arboricity $\alpha$ if the minimum number of spanning forests needed to cover the graph is $\alpha$.} The state-of-the-art clique parallel \kc counting algorithm is
\kclist~\cite{Danisch18}, which achieves the same work bound,
but does not have a strong theoretical bound on the
span (parallel time). Furthermore, \kclist as well as existing parallel \kc counting
algorithms have limited scalability for graphs with more than a few
hundred million edges, but real-world graphs today frequently contain
billions to hundreds of billions of edges~\cite{Meusel2015}.

\myparagraph{\kc Counting}
In this paper, we design a new parallel \kc counting algorithm,
\ourcount that matches the work of Chiba-Nishezeki, has
polylogarithmic span, and has improved space complexity compared to
\kclist. Our algorithm is able to significantly outperform \kclist and
other competitors, and scale to larger graphs than prior work.
\ourcount is based on using low out-degree orientations of the graph
to reduce the total work.
Assuming that we have a low out-degree ranking of the graph, we show
that for a constant $k$ we can count or list all $k$-cliques in
$\BigO{m\alpha^{k-2}}$ work, and $\BigO{k \log n + \log^2 n}$ span with high probability (\whp{}),\footnote{We say $O(f(n))$ \defn{with
    high probability (\whp{})} to indicate $O(cf(n))$ with probability
  at least $1-n^{-c}$ for $c \geq 1$, where $n$ is the input size.}
where $m$ is the number of edges in the graph and $\alpha$ is the
arboricity of the graph.  Having work bounds parameterized by $\alpha$ is
desirable since most real-world graphs have low
arboricity~\cite{DhBlSh18}. Theoretically, \ourcount requires
$\BigO{\alpha}$ extra space per processor; in contrast, the \kclist
algorithm requires $\BigO{\alpha^2}$ extra space per processor. Furthermore, \kclist does not achieve polylogarithmic span.

We also design an approximate \kc{} counting algorithm
based on counting on a sparsified graph. We show in the \intheapp{} that our approximate algorithm
produces unbiased estimates and runs in $\BigO{pm\alpha^{k-2}+m}$
work and $\BigO{k \log n + \log^2 n}$ span \whp{} for a sampling
probability of $p$. 

\myparagraph{Parallel Ranking Algorithms}
We present two new parallel algorithms
for efficiently ranking the vertices, which we use for \kc{} counting.  We show that a distributed
algorithm by Barenboim and Elkin~\cite{Barenboim10} can be implemented
in linear work and polylogarithmic span. We also parallelize an
external-memory algorithm by Goodrich and Pszona~\cite{Goodrich11} and
obtain the same complexity bounds. 
We believe that our parallel ranking algorithms may be of
independent interest, as many other subgraph finding algorithms use
low out-degree orderings
(e.g.,~\cite{Goodrich11,Pinar2017,Jain2020}).

\myparagraph{Peeling and $k$-Clique Densest Subgraph}
We also present new parallel algorithms for the
  \kcds problem, a generalization of the densest subgraph problem
that was first introduced by Tsourakakis~\cite{Tsourakakis15}. This
problem admits a natural $1/k$-approximation by peeling vertices in
order of their incident \kc counts.  We present a parallel peeling
algorithm, \ourpeel, that peels all vertices with the
lowest \kc count on each round and uses \ourcount as a subroutine.
The expected amortized work of \ourpeel is $\BigO{m\alpha^{k-2}+ \rho_k(G)\log n}$ and
the span is $\BigO{\rho_k(G) k \log n + \log^2 n}$ \whp{}, where $\rho_k(G)$ is the
number of rounds needed to completely peel the graph.
We also prove in the \intheapp{} that the problem of obtaining the hierarchy given by
this process is \pcomplete{} for $k>2$, indicating that a polylogarithmic-span 
solution is unlikely.

Tsourakakis also shows that naturally extending the Bahmani et
al.~\cite{Bahmani12} algorithm for approximate densest subgraph gives
an $1/(k(1+\epsilon))$-approximation in $\BigO{\log n}$ parallel
rounds, although they were not concerned about work.
We present
an $\BigO{m\alpha^{k-2}}$ work and polylogarithmic-span algorithm,
\ourapproxpeel, for obtaining a $1/(k(1+\epsilon))$-approximation to
the \kcds problem. We obtain this work bound using our \kc algorithm as a subroutine.
Danisch et al.~\cite{Danisch18} use their \kc counting
algorithm as a subroutine to implement these two approximation
algorithms for \kcds, but their implementations do not have
provably-efficient bounds.

\myparagraph{Experimental Evaluation}
We present implementations of our algorithms that use various
optimizations to achieve good practical performance.  We
perform a thorough experimental study on a 30-core
machine with two-way hyper-threading and compare to prior work.
We show that on a variety of real-world graphs and different 
$k$, our \kc counting algorithm achieves 1.31--9.88x speedup over the
state-of-the-art parallel \kclist algorithm~\cite{Danisch18} and
self-relative speedups of 13.23--38.99x.
We also compared our \kc counting algorithm to other
parallel \kc counting implementations including Jain and Seshadhri's
\pivoter{} \cite{Jain2020}, Mhedhbi and Salihoglu's worst-case optimal
join algorithm (\wco) \cite{Mhedhbi2019}, Lai \textit{et al.}'s
implementation of a binary join algorithm (\binaryjoin{})
\cite{BinaryJoin19}, and Pinar \textit{et al.}'s \algname{ESCAPE}
\cite{Pinar2017}, and demonstrate speedups of up to several orders of
magnitude.

Furthermore, by integrating state-of-the-art parallel graph
compression techniques, we can process graphs
with tens to hundreds of billions of edges, significantly
improving on the capabilities of existing implementations. \emph{As
  far as we know, we are the first to report $4$-clique counts for
  Hyperlink2012, the largest publicly-available graph, with over
  two hundred billion undirected edges.}

We study the accuracy-time tradeoff of our sampling algorithm, and
show that is able to approximate the clique counts with 5.05\% error
5.32--6573.63 times more quickly than running our exact counting
algorithm on the same graph. We compare our sampling algorithm to
Bressan \textit{et al.}'s serial \algname{MOTIVO} \cite{Bressan2019},
and demonstrate 92.71--177.29x speedups.  Finally, we study our two
parallel approximation algorithms for \kcds and show that our we are
able to outperform \kclist by up to 29.59x and achieve 1.19--13.76x
self-relative speedup.  We demonstrate up to 53.53x speedup over Fang
\textit{et al.}'s serial \coreapp{} \cite{FaYuChLaLi19} as well.

The contributions of this paper are as follows:

\begin{enumerate}[label=(\textbf{\arabic*}),topsep=0pt,noitemsep,parsep=0pt,leftmargin=15pt]
  \item A parallel algorithm with  $\BigO{m\alpha^{k-2}}$  and polylogarithmic 
  span \whp{} for \kc counting.
  \item  Parallel algorithms for  low out-degree
  orientations with $O(m)$ work and $O(\log^2n)$ span \whp{}.

  \item An $\BigO{m\alpha^{k-2}}$ amortized expected work parallel algorithm for
    computing a $1/k$-approximation to the \kcds
    problem, and an $\BigO{m\alpha^{k-2}}$ work and polylogarithmic-span \whp{}
    algorithm for computing a $1/(k(1+\epsilon))$-approximation.

\item Optimized implementations of our algorithms that achieve
  significant speedups over existing state-of-the-art methods, and
  scale to the largest publicly-available graphs.

\end{enumerate}

Our code is publicly available at: \url{https://github.com/ParAlg/gbbs/tree/master/benchmarks/CliqueCounting}.

\section{Preliminaries}\label{sec-par-prim}
\myparagraph{Graph Notation} We consider graphs $G =(V,E)$ to be
simple and undirected, and let $n=|V|$ and $m=|E|$. 
For any vertex $v$, $N(v)$ denotes the
neighborhood of $v$ and $\text{deg}(v)$ denotes the degree of $v$.
If there are multiple graphs, $N_G(v)$ denotes the neighborhood
of $v$ in $G$.  For a directed graph $DG$, $N(v) = N_{DG}(v)$
denotes the out-neighborhood of $v$ in $DG$.  For analysis, we assume that $m=\BigOmega{n}$.
The \defn{arboricity ($\bm{\alpha}$)} of a graph is the minimum number
of spanning forests needed to cover the graph.  $\alpha$ is
upper bounded by $\BigO{\sqrt{m}}$ and lower bounded by
$\BigOmega{1}$~\cite{ChNi85}.

A \defn{\kc} is a subgraph $G'\subseteq G$ of size $k$ where all $\binom{k}{2}$ edges are present. The \defn{\kcds} is a subgraph $G'\subseteq G$ that maximizes across all subgraphs the ratio between the number of \kc{s} induced by vertices in $G'$ and the number of vertices in $G'$~\cite{Tsourakakis15}. 
An \defn{$c$-orientation} of an undirected graph is a
total ordering on the vertices, where the oriented out-degree of each
vertex (the number of its neighbors higher than it in the ordering) is
bounded by $c$.

\myparagraph{Model of Computation} For analysis, we use the work-span
model~\cite{JaJa92,CLRS}.  The \defn{work} $W$ of an algorithm is
the total number of operations, and the \defn{span} $S$ is the longest
dependency path.  
We can execute a parallel computation in $W/P+S$ running time using
$P$ processors~\cite{Brent1974}.
We aim for \defn{work-efficient} parallel algorithms in this model,
that is, an algorithm with work complexity that asymptotically matches
the best-known sequential time complexity for the problem.  We assume
concurrent reads and writes and atomic adds are supported in the model
in $\BigO{1}$ work and span.

\myparagraph{Parallel Primitives} We use the following primitives.
\defn{Reduce-Add} takes as input a sequence $A$ of length
$n$, and returns the sum of the
entries in $A$.
\defn{Prefix sum} takes as input a sequence $A$
of length $n$, an identity $\varepsilon$, and an associative binary
operator $\oplus$, and returns the sequence $B$ of length $n$ where
$B[i] = \bigoplus_{j < i} A[j] \oplus \varepsilon$.  \defn{Filter}
takes as input a sequence $A$ of length $n$ and a predicate function
$f$, and returns the sequence $B$ containing $a \in A$ such that
$f(a)$ is true, in the same order that these entries appeared in
$A$. These primitives take $\BigO{n}$ work and $\BigO{\log n}$
span~\cite{JaJa92}.

We also use \defn{parallel integer sort}, which sorts $n$ integers
in the range $[1, n]$ in $\BigO{n}$ work \whp{} and
$\BigO{\log n}$ span \whp{}~\cite{RaRe89}.
We use \defn{parallel hash tables} that support $n$
operations (insertions, deletions, and membership queries) in $O(n)$
work and $O(\log n)$ span \whp{}~\cite{gil91a}.  Given
hash tables $\mathcal{T}_1$ and $\mathcal{T}_2$ containing $n$ and
$m$ elements respectively, the intersection $\mathcal{T}_1 \cap
\mathcal{T}_2$ can be computed in $O(\min(n,m))$ work and
$O(\log (n+m))$ span \whp{}.

\myparagraph{Parallel Bucketing} A \defn{parallel bucketing structure}
maintains a mapping from keys to buckets, which we use to group
vertices by their \kc counts in our \kcds algorithms. The
bucket value of keys can change, and the structure updates the
bucket containing these keys.

In practice, we use the bucketing structure by Dhulipala et
al.~\cite{DhBlSh17}. However, for theoretical purposes, we use the
batch-parallel Fibonacci heap by Shi 
and Shun~\cite{Shi2020}, which  
supports $b$ insertions in $O(b)$ amortized expected
work and $O(\log n)$ span \whp{}, $b$ updates in
$O(b)$ amortized work and $O(\log^{2} n)$ span \whp{}, and extracts
the minimum bucket in $O(\log n)$ amortized expected work and $O(\log
n)$ span \whp{}.

\myparagraph{Graph Storage} In our implementations, we store our
graphs in compressed sparse row (CSR) format, which requires
$\BigO{m+n}$ space.  For large graphs, we compress the edges for
each vertex using byte codes that can be decoded in
parallel~\cite{SDB2015}.
For our theoretical bounds, we  assume that graphs are represented
in an adjacency hash table, where each vertex is associated with a
parallel hash table of its neighbors.

\section{Clique Counting}\label{sec:counting}
In this section, we present our main algorithms for counting \kc{}s. 
We describe our parallel algorithm for low out-degree
orientations in Section~\ref{sec:counting:ranking}, our parallel \kc counting
algorithm in Section~\ref{sec:counting:listing}, and practical optimizations in Section~\ref{sec:counting:optimize}.
We discuss briefly our parallel approximate counting algorithm in Section~\ref{sec:counting:sampling}.

\subsection{Low Out-degree Orientation (Ranking)}\label{sec:counting:ranking}

Recall that an $c$-orientation of an undirected graph is a total
ordering on the vertices, where the oriented out-degree of each vertex
(the number of its neighbors higher than it in the ordering) is
bounded by $c$. Although this problem has been widely studied in other
contexts, to the best of our knowledge, we are not aware of any
previous work-efficient parallel algorithms for solving this problem. We show
that the Barenboim-Elkin and Goodrich-Pszona algorithms, which are
efficient in the $\mathsf{CONGEST}$ and I/O models of computation
respectively, lead to work-efficient low-span algorithms.

Both algorithms take as input a
user-defined parameter $\epsilon$. The Barenboim-Elkin algorithm also requires a
parameter, $\alpha$, which is the arboricity of the graph
(or an estimate of the arboricity). As an estimate of the arboricity, we use 
the approximate densest-subgraph algorithm from~\cite{DhBlSh18}, 
which yields a $(2+\epsilon)$-approximation and takes $O(m + n)$ work and 
$O(\log^2 n)$ span.
The algorithms peel vertices in rounds until the graph
is empty; the peeled vertices are appended to the end of
ordering.  Both algorithms peel a constant fraction
of the vertices per round. For the Goodrich-Pszona algorithm, an $\epsilon/(2+\epsilon)$ fraction of
vertices are removed on each round, 
so the algorithm finishes in $O(\log n)$ rounds.  The Barenboim-Elkin algorithm peels
vertices with induced degree less than $(2+\epsilon)\alpha$ on each
round.
By
definition of arboricity, there are at most $n\alpha/d$ vertices with
degree at least $d$.  Thus, the number of vertices with degree at
least $(2+\epsilon)\alpha$ is at most $n/(2+\epsilon)$, and a
constant fraction of the vertices have degree at most
$(2+\epsilon)\alpha$.  Since a subgraph of a graph with arboricity
$\alpha$ has arboricity at most $\alpha$, each round peels at
least a constant fraction of remaining vertices, and the algorithm
terminates in $O(\log n)$ rounds.
We provide pseudocode for the algorithms in the \intheapp{}.
 
For the $c$-orientation given by the Barenboim-Elkin algorithm,
vertices have out-degree less than $(2+\epsilon)\alpha$ by construction.
For the $c$-orientation given by the Goodrich-Pszona algorithm, the
number of vertices with degree at least $(2+\epsilon)\alpha$ is at
most $n/(2+\epsilon)$, so the $\epsilon/(2+\epsilon)$ fraction
of the lowest degree vertices must have degree less than
$(2+\epsilon)\alpha$.

We implement each round of the Goodrich-Pszona
algorithm using parallel integer sorting to find the
$\epsilon/(2+\epsilon)$ fraction of vertices with lowest induced
degree.  Our parallelization of Barenboim-Elkin uses a
parallel filter to find the set of vertices to peel. We can
implement a round in both algorithms in linear work in the number of
remaining vertices, and $\BigO{\log n}$ span.  We obtain the following
theorem, which we prove in the \intheapp{}.

\begin{restatable}[]{theorem}{goodrichefficient}
\label{thm:goodrich_efficient}
  The Goodrich-Pszona and Barenboim-Elkin algorithms compute
  $\BigO{\alpha}$-orientations  in $\BigO{m}$ work (\whp{} for Goodrich-Pszona),
  $\BigO{\log^2 n}$ span (\whp{} for Goodrich-Pszona), and $\BigO{m}$ space. 
\end{restatable}

Finally, in the rest of this paper, we direct graphs in
CSR format after computing an orientation, which can be done in
$O(m)$ work and $O(\log n)$ span using prefix sum and filter.

\subsection{Counting algorithm}\label{sec:counting:listing}

\begin{algorithm}[!t]
  \footnotesize
 \begin{algorithmic}[1]
 \Procedure{Rec-Count-Cliques}{$DG$, $I$, $\ell$}
 \State \Comment{$I$ is the set of potential neighbors to complete the clique, and $\ell$ is the recursive level}
  \If{$\ell = 1$}
   \Return $|I|$ \Comment{Base case} \label{line:base-case}
 \EndIf
 \State Initialize $T$ to store clique counts per vertex in $I$
\ParFor{$v$ in $I$}
  \State $I' \leftarrow$ \algname{intersect}($I$, $N_{DG}(v)$) \Comment{Intersect $I$ with directed neighbors of $v$} \label{line:intersect}
  \State $t' \leftarrow$ \algname{rec-count-cliques}($DG$, $I'$, $\ell-1$) \label{line:rec-count-recurse}
  \State Store $t'$ in $T$
\EndParFor
\State $t \leftarrow$ \algname{reduce-add}($T$) \Comment{Sum clique counts in $T$} \label{line:aggregate}
\State \Return $t$ \label{line:return}
\EndProcedure
\smallskip
\Procedure {\ourcount{}}{$G = (V,E)$, $k$, \algname{Orient}}
\State $DG \leftarrow$ \algname{Orient}($G$) \Comment{Apply a user-specified orientation algorithm} \label{line:orient}
\State \Return \algname{Rec-Count-Cliques}($DG$, $V$, $k$) \label{line:rec-count}
  \EndProcedure
 \end{algorithmic}
\caption{Parallel $k$-clique counting algorithm} \label{alg-count}
\end{algorithm}

 Our algorithm for $k$-clique counting is shown as \ourcount{} in
 \algprefix \ref{alg-count}.  On Line~\ref{line:orient}, \ourcount{} first directs the edges of
 $G$ such that every vertex has out-degree $O(\alpha)$, as described
 in Section \ref{sec:counting:ranking}. Then,
 it calls a recursive subroutine \algname{rec-count-cliques} that
 takes as input the directed graph $DG$, candidate vertices $I$ that
 can be added to a clique, and the number of vertices $\ell$ left to
 complete a $k$-clique (Line~\ref{line:rec-count}). With every
 recursive call to \algname{rec-count-cliques}, a new candidate vertex
 $v$ from $I$ is added to the clique and $I$ is pruned to contain only
 out-neighbors of $v$
 (Line~\ref{line:intersect}). \algname{rec-count-cliques} terminates
 when precisely one vertex is needed to complete the $k$-clique, in
 which the number of vertices in $I$ represents the number of
 completed $k$-cliques (Line~\ref{line:base-case}). The counts
 obtained from recursive calls are aggregated using a
 \algname{reduce-add} and returned
 (Lines~\ref{line:aggregate}--\ref{line:return}).

 Finally, by construction, \ourcount{} and \algname{rec-count-cliques}
can be easily modified to store $k$-clique counts per vertex.
We append \algname{-v} to indicate the corresponding subroutines 
that store counts per vertex, which are used in our peeling algorithms. 
Similarly, \ourcount{} can be modified to support \kc listing.

\myparagraph{Complexity Bounds} Aside from the initial call
to \algname{rec-count-cliques} which takes $I = V$, in
subsequent calls, the size of $I$ is bounded by $O(\alpha)$. This is
because at every recursive step, $I$ is intersected with the
out-neighbors of some vertex $v$, which is bounded by $O(\alpha)$.
The additional space required by \ourcount{} per processor is
$\BigO{\alpha}$, and since the space is allocated in a stack-allocated
fashion, we can bound the total additional space by $\BigO{P\alpha}$
on $P$ processors when using a work-stealing
scheduler~\cite{BL98}. Thus, the total space for \ourcount{} is
$\BigO{m+P\alpha}$. In contrast, the \kclist algorithm requires
$\BigO{m+P\alpha^2}$ space.

Moreover, considering the first call to \algname{rec-count-cliques},
the total work of \algname{intersect} is given by $\BigO{m}$ \whp{},
because the sum of the degrees of each vertex is bounded by
$\BigO{m}$.  Also, using a parallel adjacency hash table, the
 work of \algname{intersect} in each subsequent recursive step
is given by the minimum of $|I|$ and $|N_{DG}(v)|$, and thus is
bounded by $\BigO{\alpha}$ \whp{}. We recursively call
\algname{rec-count-cliques} $k$ times as $\ell$ ranges from $1$ to
$k$, but the first call involves a trivial intersect where
we retrieve all directed neighbors of $v$, and the final recursive call returns 
immediately with $|I|$. Hence, we have $k-2$ recursive steps that call
\algname{intersect} non-trivially, and so in total,
\ourcount{} takes $\BigO{m\alpha^{k-2}}$ work \whp{}.

The span of \ourcount{} is defined by the span of
\algname{intersect} and \algname{reduce-add} in each recursive
call. As discussed in Section~\ref{sec-par-prim}, the span of
\algname{intersect} is $\BigO{\log n}$ \whp{}, due to the use of the 
parallel hash tables, and the span
of \algname{reduce-add} is $\BigO{\log n}$. Thus, since we
have $k-2$ recursive steps with $\BigO{\log n}$ span, and taking into account 
the $\BigO{\log^2 n}$ span \whp{} in orienting the graph, \ourcount{} takes
$\BigO{k \log n + \log^2 n}$ span \whp{}.
\ourcountv{} obtains the same work and span bounds
as \ourcount{}, since the atomic add operations do
not increase the work or span.
The total complexity of $k$-clique counting is as follows.
\begin{theorem}\label{thm:clique-list}
\ourcount takes $\BigO{m\alpha^{k-2}}$ work and $\BigO{k \log n + \log^2 n}$ span \whp{}, using $\BigO{m+P\alpha}$ space on $P$ processors.
\end{theorem}

\subsection{Sampling}\label{sec:counting:sampling}
We discuss in the \intheapp{} a technique, colorful sparsification,
that allows us to produce approximate \kc
counts, based on previous work on approximate triangle and butterfly
(biclique) counting~\cite{Pagh2012, SaSaTi18}. The technique uses our
\kc counting algorithm (\algprefix~\ref{alg-count}) as a subroutine, and we prove the
following theorem in the \intheapp{}.

\begin{restatable}[]{theorem}{sampworkspan}
\label{thm:sampworkspan}
Our sampling algorithm with parameter $p=1/c$ gives an unbiased estimate of the global \kc count and takes $\BigO{pm\alpha^{k-2}+m}$ work and $\BigO{k \log n + \log^2 n}$ span \whp{}, and $\BigO{m+P\alpha}$ space on $P$ processors.
\end{restatable}

\subsection{Practical Optimizations}\label{sec:counting:optimize}
We now introduce practical optimizations that offer tradeoffs
between performance and space complexity. 
First, in the initial call to \algname{rec-count-cliques}, for each
$v$, we construct the induced subgraph on $N_{DG}(v)$
and replace $DG$ with this subgraph in later recursive
levels. Thus, later recursive levels can skip
edges that have already been pruned in the
first level. Because the out-degree of each
vertex is bounded above by $\BigO{\alpha}$, we require $\BigO{\alpha^2}$
extra space per processor to store these induced subgraphs.

Moreover, as mentioned in Section~\ref{sec-par-prim}, we store our
graphs (and induced subgraphs) in CSR format. To efficiently intersect
the candidate vertices in $I$ with the requisite out-neighbors, we
relabel vertices in the induced subgraph constructed in the second
level of recursion to be in the range $[0,\ldots,\BigO{\alpha}]$, and
then use an array of size $\BigO{\alpha}$ to mark vertices in $I$.
For each vertex $I$, we check if its out-neighbors are
marked in our array to perform \algname{intersect}.

While this would require $\BigO{k\alpha}$ extra space per processor to
maintain a size $\BigO{\alpha}$ array per recursive call, we find
that in practice, parallelizing up to the first two recursive
levels is sufficient. Subsequent recursive
calls are sequential, so we can reuse the array
between recursive calls by using the labeling scheme from Chiba and
Nishizeki's serial $k$-clique counting algorithm~\cite{ChNi85}. We
record the recursive level $\ell$ in our array for each vertex in $I$,
perform \algname{intersect} by checking if the out-neighbors have been
marked with $\ell$ in the array, and then reset the marks. 
This allows us to use only
$\BigO{\alpha}$ extra space per processor to perform
\algname{intersect} operations.

In our implementation, \defn{node parallelism} refers to parallelizing
only the first recursive level and \defn{edge parallelism} refers
to parallelizing only the first two recursive levels. These
correspond with the ideas of node and edge parallelism in Danisch et
al.'s \kclist{} algorithm~\cite{Danisch18}.  We also implemented dynamic
parallelism, where more recursive levels are parallelized, but this was slower in
practice---further parallelization did not mitigate the parallel overhead
introduced.

Finally, for the intersections on the second recursive level
 (the first set of non-trivial intersections), it is faster
in practice to use an array marking vertices in $N_{DG}(v)$. 
If we let $I_1 = N_{DG}(v)$ denote the set of 
neighbors obtained after the first recursive level, then to obtain 
the vertices in $I_2$ in the second level, we use 
a size $n$ array to mark vertices in $I_1$ and 
perform a constant-time lookup to determine for $u \in I_1$, which
out-neighbors $u' \in N_{DG}(u)$ are also in $I_1$; these $u'$ form $I_2$. 
Past the second level, 
we relabel vertices in the induced subgraph as mentioned
above and only require the $\BigO{\alpha}$ array for
 intersections. Thus, we use linear space per processor
 for the second level of recursion only.

In total, the space complexity for intersecting in the second level of
recursion and storing the induced subgraph on $N_{DG}(v)$ dominates,
and so we use $\BigO{\max (n, \alpha^2)}$ extra space per processor.

\subsection{Comparison to \kclist{}}\label{sec:counting:kclist}

\begin{figure}
  \centering
\includegraphics[width=0.65\columnwidth, page=8]{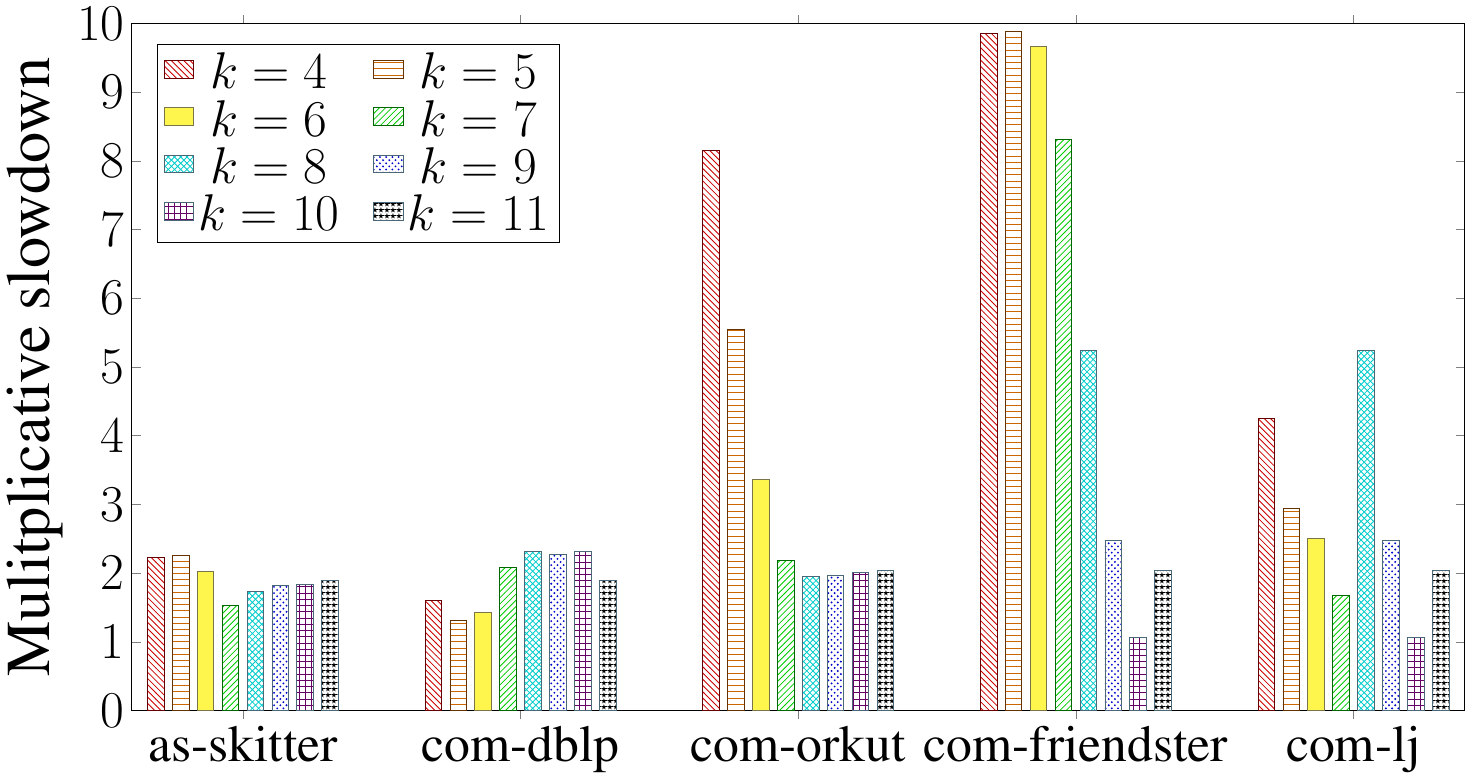}
  \caption{Multiplicative slowdowns of \kclist{}'s parallel \kc counting implementation, compared to \ourcount{}. The best runtimes between node and edge parallelism for  \kclist{} and \ourcount{}, and among different orientations for \ourcount{} are used.}\label{fig:kclist} 
\end{figure}

Some of the practical optimizations for 
\ourcount{} overlap with those in
\kclist{}~\cite{Danisch18}. 
Specifically, \kclist{} also stores the induced subgraph
on $N_{DG}(v)$, offers node and edge parallelism options, and uses a size $n$ array to mark vertices to perform
intersections. However, \ourcount{} is fundamentally different due to the low
out-degree orientation and because it does not inherently
require labels or subgraphs stored between recursive levels.

Notably, the induced subgraph that \ourcount{} computes at
the first level of recursion takes $\BigO{\alpha^2}$ space per processor because
of the low out-degree orientation, whereas \kclist{} takes
$\BigO{n^2}$ space per processor for their induced subgraph. Then, 
\ourcount{} further saves on space and computation by maintaining
only the subgraph computed from the first level of recursion to 
intersect with vertices in later recursive levels, which is solely possible
due to the low out-degree orientation, whereas \kclist{} necessarily
recomputes an induced subgraph on every recursive level.
As a result, \ourcount{} is also able to compute intersections 
using only an array of size $\BigO{\alpha}$ per recursive level, whereas \kclist{}
requires an array of size $\BigO{n}$ per level.

In total, \kclist{} uses $\BigO{n^2}$ extra space per processor,
whereas \ourcount{} uses $\BigO{\max (n, \alpha^2)}$ extra space per processor. 
Compared to \kclist, \ourcount has lower memory
footprint, span, and constant factors in the work, which allow us to
achieve speedups between 1.31--9.88x over \kclist{}'s best parallel
runtimes and which allows us to scale to the largest
publicly-available graphs, considering the best optimizations, 
as shown in Figure~\ref{fig:kclist}.
Note that for large $k$ on large graphs, the multiplicative slowdown decreases because 
\kclist{} incurs a large preprocessing overhead due to the 
large induced subgraph computed in the first recursive level, 
which is mitigated by higher 
counting times as $k$ increases. These results
are discussed further in Section~\ref{sec:counting-eval}.

\section{$k$-Clique Densest Subgraph}\label{sec:densest-subgraph}

We present our new work-efficient parallel algorithms for approximating the \kcds
problem, using the vertex peeling algorithm.

\subsection{Vertex Peeling}\label{sec:peeling:peeling}

~

\myparagraph{Algorithm} \algprefix~\ref{alg:peelexact} presents
\ourpeel, our parallel algorithm for vertex peeling, which also gives
a $1/k$-approximate to the \kcds problem.  An example of this peeling process is shown in Figure~\ref{fig:peeling-example}.
The algorithm uses
\ourcount to compute the initial per-vertex \kc counts
($C$), which are given as an argument to the algorithm.  The algorithm
first initializes a parallel bucketing structure that stores buckets
containing sets of vertices, where all vertices in the same bucket
have the same \kc count (Line~\ref{line:bucketinit}).  Then, while not
all of the vertices have been peeled, it repeatedly extracts the
vertices with the
lowest induced \kc count (Line~\ref{line:extractbucket}), updates the
count of the number of peeled vertices
(Line~\ref{line:updatefinished}), and updates the \kc counts of
vertices that are not yet finished that participate in \kc{}s with the
peeled vertices (Line~\ref{line:updatecliques}). \algname{Update} also
returns the number of \kc{}s that were removed as well as the set of vertices whose \kc counts changed.
We then update the buckets of the vertices whose \kc counts changed (Line~\ref{line:updatebuckets}).
Lastly, the algorithm
checks if the new induced subgraph has higher density than the current
maximum density, and if so updates the maximum density
(Lines~\ref{line:updatemaxdensity-if}--\ref{line:updatemaxdensity}).

The \algname{Update} procedure
(Line~\ref{line:update}--\ref{line:update-end}) performs the bulk of
the work in the algorithm. It takes each vertex in $A$ (vertices to be
peeled), builds its induced neighborhood, and counts all
$(k-1)$-cliques in this neighborhood using \ourcount,
as these $(k-1)$-cliques together with a peeled vertex form a \kc
(Line~\ref{line:peel-count-cliques}). On
Line~\ref{line:double-counting}, we avoid double counting \kc{s} by
ignoring vertices already peeled in prior rounds, and for vertices
being peeled in the same round, we first mark them in an auxiliary
array and break ties based on their rank (i.e., for a \kc involving
multiple vertices being peeled, the highest ranked vertex is
responsible for counting it).

This algorithm computes a density that approximates the
density of the $k$-clique densest subgraph.  A subgraph with this
density can be returned by rerunning the algorithm.

\begin{algorithm}[!t]
  \footnotesize
 \begin{algorithmic}[1]
\Procedure{Update}{$G=(V,E), k, DG, C, A$} \label{line:update}
 \State Initialize $T$ to store \kc counts per vertex in $A$
\ParFor{$v$ in $A$}
  \State $I \leftarrow \{ u \mid u \in N_G(v) $ and $u$ has not been previously peeled or $u \in A$ and $u \in N_{DG}(v)$ $\}$ \Comment{To avoid double counting} \label{line:double-counting}
  \State $(t',U) \leftarrow$ \algname{rec-count-cliques-v}($DG$, $I$, $k-1$, $C$) \label{line:peel-count-cliques}
   \State Store $t'$ in $T$
\EndParFor
\State $t \leftarrow$ \algname{reduce-add}($T$) \Comment{Sum \kc counts in $T$}
\State \Return $(t,U)$ \label{line:update-end}
\EndProcedure

\smallskip

  \Procedure {\ourpeel{}}{$G=(V,E), k,DG,C,t$}
  \State \Comment $C$ is an array of \kc counts per vertex and $t$ is the total \# of \kc{}s
\State Let $B$ be a bucketing structure mapping $V$ to buckets based
on \# of \kc{}s\label{line:bucketinit}
\State $d^* \leftarrow t/|V|$, $f \leftarrow 0$
\While{$f < |V|$}
\State $A \leftarrow$ vertices in next bucket in $B$ (to be peeled)\label{line:extractbucket}
\State $f \leftarrow f + |A|$\label{line:updatefinished}
\State $(t',U) \leftarrow $\algname{Update}$(G, k, DG, C, A)$ \Comment{Update \# of \kc{}s}\label{line:updatecliques}
\State Update the buckets of vertices in $U$, peeling $A$ \label{line:updatebuckets}
\If{$t'/(|V|-f) > d^*$}\label{line:updatemaxdensity-if}
\State $d^* \leftarrow t'/(|V|-f)$ \Comment{Update maximum density}\label{line:updatemaxdensity}
\EndIf
\EndWhile
\State \Return $d^*$
  \EndProcedure
 \end{algorithmic}
\caption{Parallel vertex peeling algorithm}
 \label{alg:peelexact}
\end{algorithm}

\begin{figure}[t]
  \centering
  \includegraphics[width=0.75\columnwidth, page=2]{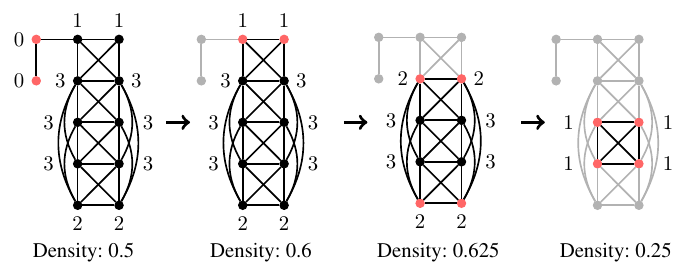}
\caption{An example of our peeling algorithm \ourpeel for $k=4$. Each vertex is labeled with its current 4-clique count. At each step, we peel the vertices with the minimum 4-clique count, highlighted in red, and then recompute the 4-clique counts on the unpeeled vertices. If there are multiple vertices with the same minimum 4-clique count, we peel them in parallel. Each step is labeled with the \kc density of the remaining graph. 
}
\label{fig:peeling-example}
\end{figure}

In the \intheapp{}, we prove that \ourpeel correctly generates a subgraph with the same approximation guarantees of 
Tsourakakis' sequential
$k$-clique densest subgraph algorithm~\cite{Tsourakakis15},
and the following bounds on the complexity of \ourpeel.
$\rho_k(G)$ is defined to be the \kc peeling complexity of
$G$, or the number of rounds needed to peel the graph where in each
round, all vertices with the minimum \kc count are peeled.  Note that
$\rho_{k}(G) \leq n$.  The proof requires applying bounds from the
batch-parallel Fibonacci heap~\cite{Shi2020} and using the
Nash-Williams theorem~\cite{nash1961edge}.

\begin{restatable}[]{theorem}{peelexact}
\label{thm:peelexact}
  \ourpeel computes a $1/k$-approximation to the
  $k$-clique densest subgraph problem in $\BigO{m\alpha^{k-2} +
  \rho_k(G)\log n}$ expected amortized work, $\BigO{\rho_k(G) k
  \log n + \log^2 n}$ span \whp{}, and $\BigO{m + P \alpha}$ space, where
  $\rho_k(G)$ is the \kc peeling complexity of $G$.
\end{restatable}

\myparagraph{Discussion}
To the best of our knowledge, Tsourakakis presents the first sequential
algorithm for this problem, although the work bound is worse than ours in most cases.
Sariyuce et al.~\cite{SariyuceP18} present a sequential algorithm for a more
general problem, but in the case that is equivalent to \kc peeling,
their fastest algorithm runs in $O(R(G,k))$
work and $O(C(G,k))$ space, where $R(G,k)$ is the cost of an arbitrary
\kc counting algorithm and $C(G,k)$ is the number of $k$-cliques
in $G$.
They provide another algorithm which
runs in $O(m+n)$ space, but requires $O(\sum_{v} d(v)^{k})$ work,
which could be as high as $O(n^{k})$.
Our sequential bounds are asymptotically better
than theirs in terms of either work or space, except
in the highly degenerate case where $C(G,k) =
o(\rho\log n)$.
Sariyuce
et al.~\cite{SaSePi18} also give a parallel algorithm, which is similarly not work-efficient.

\subsection{Approximate Vertex Peeling}\label{sec:peeling:approx}
We present a $1/(k(1+\epsilon))$-approximate algorithm
\ourapproxpeel{} for the \kcds problem based on approximate
peeling. The algorithm is similar to 
\ourpeel, but in each round, it sets a threshold $t = k(1+\epsilon)\tau(S)$ where $\tau(S)$ is the
density of the current subgraph $S$, and removes all vertices with
at most $\tau$ \kc{}s. Tsourakakis~\cite{Tsourakakis15}
describes this procedure and shows that it computes a
$1/(k(1+\epsilon))$-approximation of the \kcds in $O(\log n)$ rounds. 
Although the round complexity in Tsourakakis' implementation
is low, no non-trivial bound was known for its work. 
\ourapproxpeel is similar to Tsourakakis' algorithm,
except we utilize the fast, parallel $k$-clique counting methods
introduced in this paper. We prove the following in the \intheapp{}.

\begin{restatable}[]{theorem}{approxpeeling}
\label{thm:approx-peeling}
\ourapproxpeel computes a $1/(k(1+\epsilon))$-approximation to the
\kcds and runs in $\BigO{m\alpha^{k-2}}$ work and
$\BigO{k \log^2 n}$ span \whp{}, and $\BigO{m + P \alpha}$ space.
\end{restatable}

Note that the span for \ourapproxpeel{} matches or improves upon
that for \ourpeel{}; notably, when $\rho_k(G) = o(\log n)$, then
\ourapproxpeel{} takes $\BigO{\rho_k(G) k
  \log n + \log^2 n}$ span \whp{}, which is better than what is stated in Theorem~\ref{thm:approx-peeling}.

\subsection{Practical Optimizations}\label{sec:peeling:optimize}
We use the same optimizations described in Section
\ref{sec:counting:optimize} for updating \kc counts.
Also, we use the bucketing structure given by Dhulipala et
al.~\cite{DhBlSh17}, which keeps buckets
relating \kc counts to vertices, but only
materializes a constant number of the lowest buckets. 
If large ranges of buckets contain no vertices, this structure
skips over such ranges, allowing for fast retrieval of
vertices to be peeled in every round using linear space.

\section{Experiments}\label{sec:eval}

\begin{table}[t]
\small
\centering
\setlength{\tabcolsep}{2pt}

\scalebox{0.9}{
\begin{tabular}{lll}
\toprule
& $n$         & $m$                  \\ \midrule
\textbf{com-dblp}~\cite{SNAP}. &317,080   & 1,049,866  \\ \hline
\textbf{com-orkut}~\cite{SNAP}.& 3,072,441 & 117,185,083 \\ \hline
\textbf{com-friendster}~\cite{SNAP}.& 65,608,366 & $1.806\times 10^9$  \\ \hline
\textbf{com-lj}~\cite{SNAP}.&3,997,962 & 34,681,189  \\ \hline
\textbf{ClueWeb}~\cite{Lemur}& 978,408,098 & $7.474 \times 10^{10}$ \\ \hline
\textbf{Hyperlink2014}~\cite{Meusel2015}& $1.725\times 10^9$ & $1.241\times 10^{11}$ \\ \hline
\textbf{Hyperlink2012}~\cite{Meusel2015}&  $3.564 \times 10^9$ &  $2.258\times 10^{11}$ \\ \hline
\end{tabular}
}
\caption{Sizes of our input graphs. ClueWeb, Hyperlink2012, and Hyperlink2014
are symmetrized to be undirected graphs, and are stored and read in a
compressed format from the Graph Based Benchmark Suite
(GBBS)~\cite{DhBlSh18}.
}\label{table:graph-stats}
\end{table}

\begin{table*}[!h]
  \small
\centering
\begin{tabular}{p{1.25cm}llllllllll}
\toprule
&   & $k=4$      & $k=5$     & $k=6$     & $k=7$       & $k=8$      & $k=9$        & $k=10$        & $k=11$       \\ \midrule

\textbf{com-} & \ourcount{} $T_{60}$\hspace{0.1cm} & \tblsty{0.10}  & \tblsty{0.13}   & \tblsty{0.30}  & \tblsty{2.05${}^e$}   & 24.06${}^e$ & 281.39${}^e$  & 2981.74${}^{*e}$ & $> 5$ hrs    \\
\textbf{dblp} & \ourcount{} $T_{1}$\hspace{0.1cm}  &    1.57 &  1.71 &  5.58 &  64.27 &  837.82 &  9913.01 &  $> 5$ hrs &  $> 5$ hrs      \\
& \kclist{} $T_{60}$\hspace{0.1cm}                 & 0.16     & 0.17    & $0.43^e$    & $4.28^e$      & $55.78^e$    & $640.48^e$     & $6895.16^e$     & $> 5$ hrs    \\
& \pivoter{} $T_{60}$  \hspace{0.1cm}  &   2.88 &         2.88     & 2.88     & 2.88 &     \tblsty{2.88}   & \tblsty{2.88 }   & \tblsty{2.88}   & \tblsty{2.88} \\
& \wco{} $T_{60}$  \hspace{0.1cm}        & 0.19      &0.37        &3.84      & 66.06        &1126.69     &9738.00     &  $> 5$ hrs     & $> 5$ hrs  \\
& \binaryjoin{} $T_{60}$  \hspace{0.1cm}  &   0.12 &          0.42      & 2.08     & 39.29        & 627.48    & 7282.79    & $> 5$ hrs    & $> 5$ hrs \\
\hline

\textbf{com-} & \ourcount{} $T_{60}$\hspace{0.1cm} & \tblsty{3.10} & \tblsty{4.94} & \tblsty{12.57}  & \tblsty{42.09} & \tblsty{150.87${}^\circ$}  & \tblsty{584.39${}^\circ$}  & 2315.89${}^\circ$ & 8843.51${}^{\circ e}$ \\
\textbf{orkut}& \ourcount{} $T_{1}$\hspace{0.1cm}  &  79.62 & 158.74 & 452.47  &  1571.49  & 5882.83 & $> 5$ hrs & $> 5$ hrs  & $> 5$ hrs            \\
& \kclist{} $T_{60}$ \hspace{0.1cm}       & 25.27    & 27.40  & 42.23    & $91.67^e$     & $293.92^e$   & $1147.50^e$    & $4666.03^e$     & $> 5$ hrs    \\
& \pivoter{} $T_{60}$  \hspace{0.1cm}  &   292.35 &        385.04    & 462.05    &517.29 &     559.75   & 598.88   & \tblsty{647.18}   &  \tblsty{647.18}\\
& \wco{} $T_{60}$  \hspace{0.1cm}                & 10.71        & 50.51        &267.47     &1398.89   &  6026.99     &$> 5$ hrs    & $> 5$ hrs    & $> 5$ hrs  \\
& \binaryjoin{} $T_{60}$  \hspace{0.1cm}  &    12.74 &         29.09      & 93.06     & 413.50     & 1938.06         & 9732.86    & $> 5$ hrs     & $> 5$ hrs \\
\hline

\textbf{com-}& \ourcount{} $T_{60}$ \hspace{0.1cm}& \tblsty{109.46}  &\tblsty{ 111.75} & \tblsty{115.52 }& \tblsty{139.98 }  & \tblsty{300.62}  & \tblsty{1796.12${}^e$} & \tblsty{16836.41${}^{\circ e}$} & $> 5$ hrs    \\
\textbf{friendster} & \ourcount{} $T_{1}$ \hspace{0.1cm} &     2127.79& 2328.48& 2723.53& 3815.24& 8165.76& $> 5$ hrs& $> 5$ hrs& $> 5$ hrs       \\
& \kclist{} $T_{60}$ \hspace{0.1cm}                & 1079.22  & 1104.28 & 1117.31  & 1162.84   & $1576.61^e$  & $4449.81^e$    & $> 5$ hrs     & $> 5$ hrs    \\
& \wco{} $T_{60}$  \hspace{0.1cm}                &    201.82         & 379.59      &1001.52     &4229.20    &  $> 5$ hrs  &  $> 5$ hrs  &  $> 5$ hrs  &  $> 5$ hrs \\
& \binaryjoin{} $T_{60}$  \hspace{0.1cm}   &    163.90 &          212.53      & 221.93     & 632.40     & 4532.60          &  $> 5$ hrs   & $> 5$ hrs     & $> 5$ hrs \\
 \hline

\textbf{com-lj} & \ourcount{} $T_{60}$ \hspace{0.1cm}& \tblsty{1.77}   & \tblsty{7.52}   & \tblsty{258.46} & \tblsty{10733.21 }& $> 5$ hrs  & $> 5$ hrs    & $> 5$ hrs     & $> 5$ hrs    \\
& \ourcount{} $T_{1}$\hspace{0.1cm}  & 33.04  & 231.15  & 8956.53  & $> 5$ hrs  & $> 5$ hrs  & $> 5$ hrs  & $> 5$ hrs  & $> 5$ hrs     \\
& \kclist{} $T_{60}$ \hspace{0.1cm}               & 7.53     & 22.13  & $647.77^e$  & $> 5$ hrs   & $> 5$ hrs  & $> 5$ hrs    & $> 5$ hrs     & $> 5$ hrs    \\
& \pivoter{} $T_{60}$  \hspace{0.1cm}  & 268.06  &  1475.99    & 7816.13   &$> 5$ hrs &  $> 5$ hrs    & $> 5$ hrs  &  $> 5$ hrs & $> 5$ hrs\\
& \wco{} $T_{60}$  \hspace{0.1cm}            &  6.62       & 80.78     &3448.70    & $> 5$ hrs   & $> 5$ hrs  & $> 5$ hrs    & $> 5$ hrs     & $> 5$ hrs\\
& \binaryjoin{} $T_{60}$   \hspace{0.1cm} &    4.10 &         42.32    &1816.87     & $> 5$ hrs     & $> 5$ hrs         & $> 5$ hrs     & $> 5$ hrs     & $> 5$ hrs  \\

\end{tabular}
\caption{Best runtimes in seconds for our parallel ($T_{60}$) and single-threaded ($T_1$) $k$-clique counting algorithm (\ourcount{}), as well as the best parallel runtimes from \kclist{}~\protect\cite{Danisch18}, \pivoter{}~\protect\cite{Jain2020}, \wco{}~\protect\cite{Mhedhbi2019}, and \binaryjoin{}~\protect\cite{BinaryJoin19}.
Note that we cannot report runtimes from \pivoter{} for the com-friendster graph, because for all $k$, \pivoter{} runs out of 
memory and is unable to complete $k$-clique counting.
The fastest runtimes for each experiment are bold and in green. All runtimes are from tests in the same computing environment, and include time spent preprocessing and counting (but not time spent loading the graph). 
For our parallel and serial runtimes and \kclist{}, we have chosen the fastest orientations and choice between node and edge parallelism per experiment. For the runtimes from \ourcount{}, we have noted the orientation used; ${}^\circ$ refers to the Goodrich-Pszona orientation, ${}^*$ refers to the orientation given by $k$-core, and no superscript refers to the orientation given by degree ordering. For the runtimes from \ourcount{} and \kclist{}, we have noted whether node or edge parallelism was used; ${}^e$ refers to edge parallelism, and no superscript refers to node parallelism.} \label{table:count}
\end{table*}

\begin{figure}[t]
    \centering
   \begin{subfigure}{.38\columnwidth}
   \centering
   \includegraphics[width=\columnwidth, page=1]{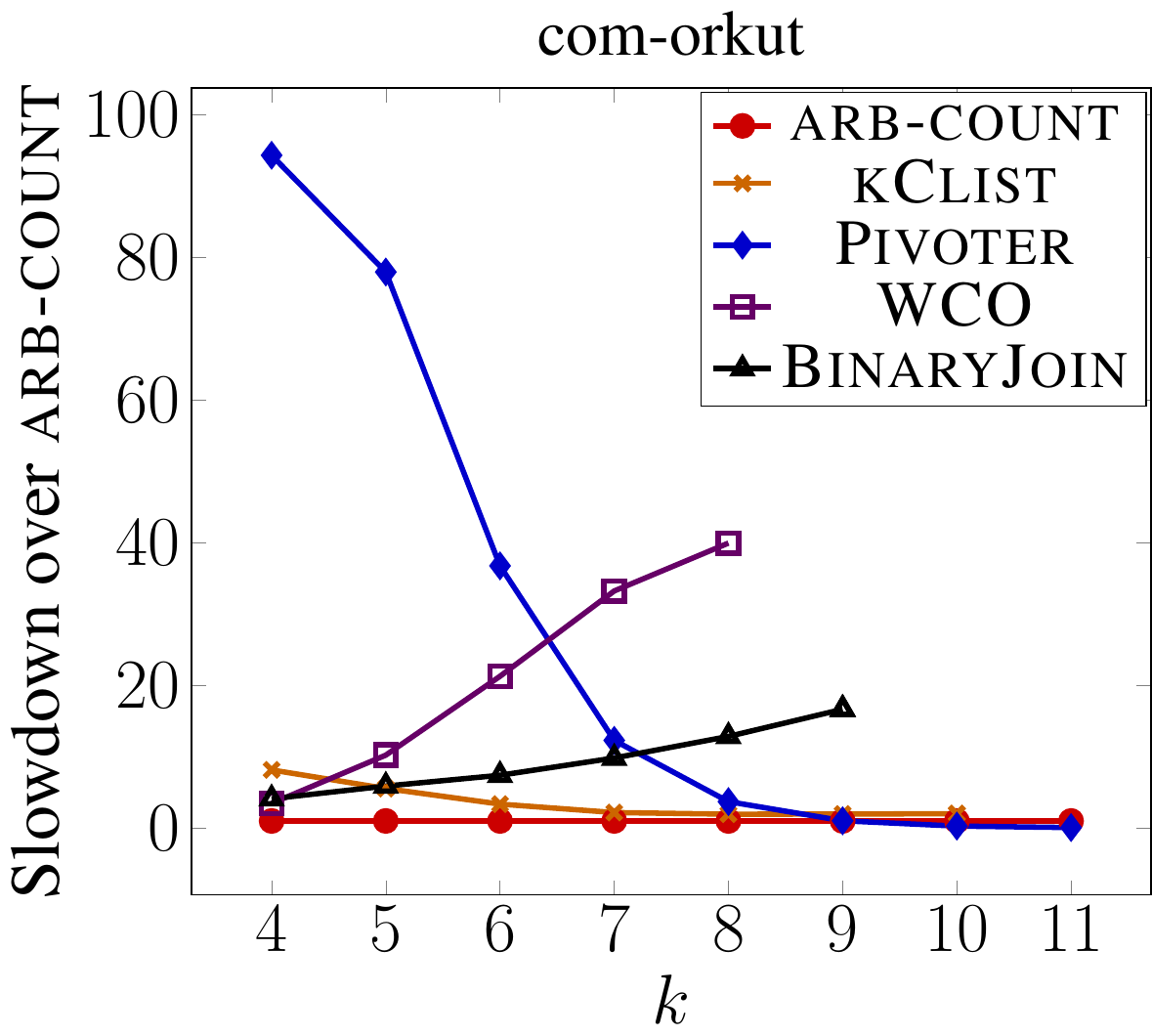}
   \end{subfigure}%
   \hfill
   \begin{subfigure}{.38\columnwidth}
   \centering
   \includegraphics[width=\columnwidth, page=2]{images/fig4.pdf}
      \end{subfigure}
   \caption{Multiplicative slowdowns of various parallel \kc counting implementations, compared to \ourcount, on com-orkut and com-friendster. The best runtimes for each implementation were used, and we have excluded any running time over 5 hours for \wco and \binaryjoin.
     Note that \pivoter was unable to perform \kc counting on com-friendster due to memory limitations, and as such is not included in this figure.
   }
    \label{fig:count-orkut-friendster}
\end{figure}

\begin{table*}[!h]
  \small
\centering
\begin{tabular}{lllllllll}
\toprule
& & $k=3$ & $k=4$                            & $k=5$   & $k=6$          & $k=7$        & $k=8$           & $k=9$                       \\ \midrule

\textbf{com-dblp} & \ourpeel{} $T_{60}$\hspace{0.1cm} & 0.14& \tblsty{0.21 }   & \tblsty{0.23${}^{\circ}$} & \tblsty{1.29${}^\circ$} & \tblsty{18.77}           & \tblsty{276.69${}^\circ$ } & \tblsty{3487.09${}^\circ$ }\\
& \ourpeel{} $T_{1}$\hspace{0.1cm}  & 0.27 &    0.37   & 1.378 & 17.99 & 258.24 & 3373.05 & $> 5$ hrs         \\
& \kclist{} $T_{1}$ \hspace{0.1cm}         &    0.19      & 0.25    & 1.10          & 14.98        & 221.98          & 2955.87         & $> 5$ hrs       \\ 
& \coreapp{} $T_{1}$\hspace{0.1cm} & \tblsty{ 0.10 } & 0.23 & 1.09        & 12.21    & 244.81      & 7674.55      & $> 5$ hrs \\  \hline

\textbf{com-orkut} & \ourpeel{} $T_{60}$\hspace{0.1cm} & \tblsty{ 33.15${}^{\circ}$} & \tblsty{76.91}   & \tblsty{221.28}         & \tblsty{721.73}       & \tblsty{2466.99$^\circ$} & \tblsty{9062.99${}^\circ$} & $> 5$ hrs       \\
& \ourpeel{} $T_{1}$  \hspace{0.1cm}&  130.04 &   184.28& 422.20  & 1032.19& 3123.72& $> 5$ hrs& $> 5$ hrs         \\
& \kclist{} $T_{1}$\hspace{0.1cm}      &    87.71          & 218.94  & 587.24         & 2029.43      & 7414.77         & $> 5$ hrs       & $> 5$ hrs       \\  
& \coreapp{} $T_{1}$\hspace{0.1cm} & 113.27 & 546.13 & 2460.65      & 16320.24 & $> 5$ hrs     & $> 5$ hrs    & $> 5$ hrs \\ \hline

\textbf{com-friendster} & \ourpeel{} $T_{60}$\hspace{0.1cm} &  \tblsty{ 371.52} &    \tblsty{ 1747.92}& \tblsty{4144.96 }& \tblsty{6870.06} & $> 5$ hrs & $> 5$ hrs & $> 5$ hrs        \\
& \ourpeel{} $T_{1}$\hspace{0.1cm}  & 3297.14&   11540.73  & 12932.28   &   14112.95  &      $> 5$ hrs            &       $> 5$ hrs           &      $> 5$ hrs            \\
& \kclist{} $T_{1}$ \hspace{0.1cm}        &   2225.70       & 3216.92 & 4325.73        & 6933.32      & $> 5$ hrs       & $> 5$ hrs       & $> 5$ hrs       \\  
& \coreapp{} $T_{1}$\hspace{0.1cm} & $> 5$ hrs &  $> 5$ hrs  &  $> 5$ hrs     & $> 5$ hrs  & $> 5$ hrs     & $> 5$ hrs    & $> 5$ hrs  \\ \hline

\textbf{com-lj} & \ourpeel{} $T_{60}$\hspace{0.1cm} & \tblsty{ 6.46 } &  \tblsty{ 26.36}& \tblsty{324.77}& \tblsty{12920.08}& $> 5$ hrs& $> 5$ hrs& $> 5$ hrs        \\
& \ourpeel{} $T_{1}$\hspace{0.1cm}  &  17.74 &   70.12  & 822.10 & $> 5$ hrs        &      $> 5$ hrs    &      $> 5$ hrs    &   $> 5$ hrs          \\
& \kclist{} $T_{1}$\hspace{0.1cm}           & 16.64        & 42.16   & 839.13         & $> 5$ hrs    & $> 5$ hrs       & $> 5$ hrs       & $> 5$ hrs      \\
& \coreapp{} $T_{1}$\hspace{0.1cm} & 7.20 & 27.53 & 1595.04     & $> 5$ hrs  & $> 5$ hrs     & $> 5$ hrs    & $> 5$ hrs 

\end{tabular}
\caption{Best runtimes in seconds for our parallel and single-threaded $k$-clique peeling algorithm (\ourpeel{}), as well as the best sequential runtimes from previous work (\kclist{} and \coreapp{})~\protect\cite{Danisch18, FaYuChLaLi19}. \kclist{} and \coreapp{} do not have parallel implementations of $k$-clique peeling; they are only serial. The fastest runtimes for each experiment are bolded and in green. All runtimes are from tests in the same computing environment, and include only time spent peeling. For our parallel runtimes, we have chosen the fastest orientations per experiment, while for our serial runtimes, we have fixed the degree orientation. For the parallel runtimes from \ourpeel{}, we have noted the orientation used; ${}^\circ$ refers to the Goodrich-Pszona orientation, and no superscript refers to the orientation given by degree ordering.}\label{table:peel}
\end{table*}

\myparagraph{Environment}
We run most of our experiments on a machine with
30 cores (with two-way hyper-threading),
with 3.8GHz Intel Xeon Scalable (Cascade Lake) processors and 240 GiB
of main memory. For our large compressed graphs, we use a machine with
80 cores (with two-way hyper-threading),
with 2.6GHz Intel Xeon E7 (Broadwell E7) processors and 3844 GiB of
main memory. We compile our programs with g++ (version 7.3.1)
using the \texttt{-O3} flag. We use OpenMP for our $k$-clique counting runtimes, and we use a lightweight scheduler called Homemade
for our $k$-clique peeling runtimes~\cite{BlAnDh20}.
We terminate
any experiment that takes over 5 hours, except for experiments on
the large compressed graphs.

\myparagraph{Graph Inputs} We test our algorithms on real-world
graphs from the Stanford Network Analysis Project
(SNAP)~\cite{SNAP}, CMU's Lemur
project~\cite{Lemur}, and the WebDataCommons
dataset~\cite{Meusel2015}. The details of the graphs are in Table~\ref{table:graph-stats},
and we show additional statistics in the \intheapp{}.

\myparagraph{Algorithm Implementations} We test different
orientations for our counting and peeling algorithms, including the
Goodrich-Pszona and Barenboim-Elkin orientations from Section
\ref{sec:counting:ranking}, with $\varepsilon =1$. We
also test other orientations that do not give work-efficient
and polylogarithmic-span bounds, but are fast in
practice, including the orientation given by ranking vertices by
non-decreasing degree, the orientation given by the $k$-core
ordering~\cite{MaBe83}, and the orientation given by the original
ordering of vertices in the graph.

Moreover, we compare our algorithms against
\kclist{}~\cite{Danisch18}, which contains state-of-the-art parallel
and sequential \kc counting algorithms, and sequential \kc peeling
implementations. \kclist{} additionally includes a parallel
approximate \kc peeling implementation. We include a simple
modification to their \kc counting code to support faster $k$-clique
counting, where we simply return the number of 
$k$-cliques instead of iterating over each $k$-clique in the final level of recursion. 
\kclist{} also
offers the option of node or edge parallelism, but only offers a
$k$-core ordering to orient the input graphs. Note that
\kclist{} does not offer a choice of orientation.

We additionally compare our counting algorithms to Jain and
Seshadhri's \pivoter{} algorithm~\cite{Jain2020}, Mhedhbi and
Salihoglu's worst-case optimal join algorithm
(\wco)~\cite{Mhedhbi2019}, Lai \textit{et al.}'s implementation of a
binary join algorithm (\binaryjoin{})~\cite{BinaryJoin19}, and Pinar
\textit{et al.}'s \algname{ESCAPE} algorithm~\cite{Pinar2017}.  Note
that \pivoter{} is designed for counting all cliques, and the latter
three algorithms are designed for general subgraph counting.  Finally, we
compare our approximate \kc counting algorithm to Bressan \textit{et
  al.}'s \algname{MOTIVO} algorithm for approximate subgraph
counting~\cite{Bressan2019}, which is more general. For \kc peeling,
we compare to Fang \textit{et al.}'s \coreapp{}
algorithm~\cite{FaYuChLaLi19} and Tsourakakis's~\cite{Tsourakakis15}
triangle densest subgraph implementation.

\subsection{Counting Results}\label{sec:counting-eval}
Table~\ref{table:count} shows the best parallel 
runtimes for $k$-clique counting over the SNAP datasets, from \ourcount, \kclist{}, \pivoter{}, \wco{}, and \binaryjoin{}, considering different orientations for \ourcount, and considering
node versus edge parallelism for \ourcount and for \kclist{}. 
We also show the best sequential runtimes from \ourcount{}.
We do not include triangle counting results, because for triangle
counting, our $k$-clique counting algorithm becomes precisely Shun and
Tangwongsan's~\cite{ShunT2015} triangle counting
algorithm. Furthermore, we performed experiments on \algname{ESCAPE}
by isolating their 4- and 5-clique counting code, but \kclist{} consistently
outperforms \algname{ESCAPE}; thus, we have not included \algname{ESCAPE}
in Table~\ref{table:count}. Figure~\ref{fig:count-orkut-friendster} shows the slowdowns of the parallel 
implementations over \ourcount on com-orkut and com-friendster.

We also obtain parallel runtimes for $k = 4$ on 
large compressed graphs, using degree ordering and node parallelism, 
on a 80-core machine with hyper-threading; note that 
\kclist{}, \pivoter{}, \wco{}, and \binaryjoin{} cannot handle these graphs. 
The runtimes are: 5824.76 seconds on ClueWeb with 74 billion 
edges ($< 2$ hours), 12945.25 seconds on Hyperlink2014 with over 
one hundred billion edges ($< 4$ hours), and 161418.89 seconds on 
Hyperlink2012 with over two hundred billion edges ($< 45$ hours). As 
far as we know, these are the first results for $4$-clique counting for graphs of this scale.

Overall, on 30 cores, \ourcount obtains speedups between 1.31--9.88x
over \kclist{}, between 1.02--46.83x over \wco{}, and between
1.20-28.31x over \binaryjoin{}.  Our largest speedups are for large
graphs (e.g., com-friendster) and for moderate values of $k$,
because we obtain more parallelism relative to the necessary work.

Comparing our parallel runtimes to \kclist{}'s serial runtimes (which
were faster than those of \wco{} and \binaryjoin{}), we obtain between
2.26--79.20x speedups, and considering only parallel runtimes over
0.7 seconds, we obtain between 16.32--79.20x speedups.
By virtue of our
orientations, our single-threaded runtimes are often faster than
the serial runtimes of the other implementations, with up to 23.17x
speedups particularly for large graphs and large values of $k$. Our
self-relative parallel speedups are between 13.23--38.99x.

We also compared with \pivoter{}~\cite{Jain2020}, which is designed
for counting all cliques, but can be truncated for fixed $k$. 
Their algorithm is able to count all cliques for 
com-dblp and com-orkut in under 5 hours. 
 However, their algorithm is not
theoretically-efficient for fixed $k$, taking $\BigO{n\alpha^2 3^{\alpha / 3}}$ work, and as such their parallel
implementation is up to 196.28x slower compared to parallel
\ourcount{}, and
their serial implementation is up to 184.76x slower
compared to single-threaded \ourcount{}.  These slowdowns are
particularly prominent for small $k$. 
Also, \pivoter{}'s truncated algorithm does not give significant 
speedups over their full algorithm, and \pivoter requires
significant space and runs out of memory 
for large graphs; it is unable to compute $k$-clique counts at all
for $k \geq 4$ on com-friendster.

Of the different orientations, using degree ordering is generally the
fastest for small $k$ because it requires little 
overhead and gives sufficiently low out-degrees. However, for
larger $k$, this overhead is less significant compared to the
time for counting and other orderings result in faster counting. The
cutoff for this switch occurs generally at $k=8$. 
Note that the Barenboim-Elkin and original orientations are never the fastest
orientations. The
slowness of the former is because it gives
a lower-granularity ordering, since it does not order between vertices
deleted in a given round. 
We found that the self-relative speedups of orienting the
graph alone were between 6.69--19.82x across all orientations, the
larger of which were found in large graphs.
We discuss preprocessing
overheads in more detail in the \intheapp{}.

Moreover, in both \ourcount{} and \kclist{}, node parallelism is
faster on small $k$, while edge parallelism is faster on large
$k$. This is because parallelizing the first level of recursion is
sufficient for small $k$, and edge parallelism introduces greater
parallel overhead. For large $k$, there is more work, which edge
parallelism balances better, and the additional parallel overhead is
mitigated by the balancing. The
cutoff for when edge parallelism is generally faster than node
parallelism occurs around $k=8$.  We provide more detailed
analysis in the \intheapp{}.

We also evaluated our approximate counting algorithm on com-orkut and com-friendster,
and compared to \algname{MOTIVO}~\cite{Bressan2019}. We defer 
a detailed discussion to the \intheapp{}. Overall, we obtain significant speedups over exact \kc counting
and have low error rates over the
exact global counts, with between
5.32--2189.11x speedups over exact counting and between 0.42--5.05\% error. We also
see 92.71--177.29x speedups over \algname{MOTIVO} for 4-clique and 5-clique approximate counting on
com-orkut.

\subsection{Peeling Results}

Table~\ref{table:peel} shows the best parallel and sequential runtimes
for $k$-clique peeling on SNAP datasets for \ourpeel,
\kclist{}, and \coreapp{} (\kclist{} and \coreapp{} only implement
sequential algorithms for exact $k$-clique peeling).

Overall, our parallel implementation obtains between 1.01--11.83x
speedups over \kclist{}'s serial runtimes. The higher speedups occur
in graphs that require proportionally fewer parallel peeling rounds
$\rho_k$ compared to its size; notably, com-dblp requires few parallel
peeling rounds, and we see between 4.78--11.83x speedups over
\kclist{} on com-dblp for $k\geq 5$. As such, our parallel speedups
are constrained by $\rho_k$.
Similarly, we obtain up to 53.53x
speedup over \coreapp{}'s serial runtimes. \coreapp{} outperforms our
parallel implementation on triangle peeling for 
com-dblp, again owing to the proportionally fewer parallel peeling
rounds in these cases.
\ourpeel achieves self-relative parallel speedups between 1.19--13.76x.
Our single-threaded runtimes are generally
slower than \kclist{}'s and \coreapp{}'s sequential runtimes owing to the parallel
overhead necessary to aggregate $k$-clique counting updates between
rounds.
In the \intheapp{}, we present a further analysis of the distributions of number of vertices peeled per round.

Moreover, the edge density of the approximate \kcds found by \ourpeel 
converges towards 1 for $k\geq 3$, and as such, \ourpeel is able to 
efficiently find large subgraphs that approach cliques. In particular, 
the \kcds that \ourpeel finds on com-lj contains 386 
vertices with an edge density of 0.992. Also, the \kcds
that \ourpeel finds on com-friendster contains 141 vertices with an
edge density of 0.993.

We also tested Tsourakakis's~\cite{Tsourakakis15} triangle densest
subgraph implementation; however, it requires
too much memory to run for com-orkut, com-friendster, and
com-lj on our machines. It completes $3$-clique peeling on com-dblp in
0.86 seconds, while our parallel \ourpeel{} takes 0.27 seconds.

Finally, we compared our parallel approximate
\ourapproxpeel to \kclist's parallel approximate algorithm 
on com-orkut and com-friendster.
\ourapproxpeel is up to 29.59x
faster than \kclist{} for large $k$, and we see between
5.95--80.83\% error on the maximum $k$-clique density obtained 
compared to the density obtained from $k$-clique peeling.

\section{Related Work}

\myparagraph{Theory} A trivial algorithm can compute all $k$-cliques
in $\BigO{n^k}$ work. Using degree-based thresholding enables
clique counting in $\BigO{m^{k/2}}$ work, which is asymptotically
faster for sparse graphs. Chiba and Nishizeki give an
algorithm with improved complexity for sparse graphs,
in which all \kc{}s can be found in $\BigO{m\alpha^{k-2}}$
work~\cite{ChNi85}, where $\alpha$ is the arboricity of the graph.

For arbitrary graphs, the fastest theoretical algorithm uses matrix
multiplication, and counts $3l$ cliques in $\BigO{n^{l\omega}}$ time
where $\omega$ is the matrix multiplication
exponent~\cite{nevsetvril1985complexity}. The \kc problem is a
canonical hard problem in the FPT literature, and is known to be
$W[1]$-complete when parametrized by $k$~\cite{downey1995fixed}.  We
refer the reader to~\cite{vassilevska2009efficient}, which surveys
other theoretical algorithms for this problem.

Recent work by Dhulipala et al.~\cite{DhulipalaLSY21} studied \kc
counting in the parallel batch-dynamic setting. One of their
algorithms calls our \ourcount{} as a
subroutine.

\myparagraph{Practice} The special case of counting and listing
triangles ($k=3$) has received a huge amount of attention over the
past two decades (e.g.,~\cite{TsDrMiKoFa11, Tsourakakis11a,
ShunT2015,  Pagh2012}, among many
others).  Finocchi et al.~\cite{Finocchi2015} present parallel \kc
counting algorithms for MapReduce.  Jain and Seshadri~\cite{Jain2017}
provide algorithms for estimating \kc counts.  The state-of-the-art
\kc counting and listing algorithm is \kclist by Danisch
et al.~\cite{Danisch18}, which is based on the
Chiba-Nishizeki algorithm, but uses the $k$-core 
ordering (which is not parallel) to rank vertices.
It achieves $\BigO{m\alpha^{k-2}}$ work, but
does not have polylogarithmic span due to the ordering
and only parallelizing one or two levels of recursion. 
Concurrent with our work, Li et al.~\cite{Li2020} present an ordering heuristic for \kc counting based on graph coloring, which they show improves upon \kclist in practice. It would be interesting in the future to study their heuristic applied to our algorithm.

Additionally, many algorithms have been designed for finding 4- 
and 5-vertex subgraphs
(e.g.,~\cite{Pinar2017,Park2018,AhmedNRDW17,Wang2018,Rossi0A19})
as well as estimating larger subgraph counts
(e.g.,~\cite{BressanCKLP18,Bressan2019}), and these algorithms can be
used for counting exact or approximate \kc counting as a special case.
Worst-case optimal join algorithms from the database
literature~\cite{AbergerLTNOR17,Ngo2018,Mhedhbi2019, BinaryJoin19} can also be used
for \kc listing and counting as a special case, and would require
$\BigO{m^{k/2}}$ work.

Very recently, Jain and Seshadri~\cite{Jain2020} present a sequential
and a vertex parallel \pivoter algorithm for counting all cliques in a
graph.
However, their algorithm cannot be used for
\kc listing as they avoid processing all cliques, and requires much
more than $\BigO{m\alpha^{k-2}}$ work in the worst case.

\myparagraph{Low Out-degree Orientations} A canonical technique in the
graph algorithms literature on clique counting, listing, and related
tasks~\cite{eppstein2010listing,Jain2020,Pinar2017} is the
use of a low out-degree orientation. Matula and Beck~\cite{MaBe83} show that $k$-core 
gives an $\BigO{\alpha}$ orientation.
However, the problem of computing this ordering is
\pcomplete{}~\cite{anderson84pcomplete}, and thus unlikely to have
polylogarithmic span.  More recent work in the distributed and
external-memory literature has shown that such orderings can be
efficiently computed in these settings. Barenboim and Elkin give a
distributed algorithm that finds an $O(\alpha)$-orientation in $O(\log
n)$ rounds~\cite{Barenboim10}.  Goodrich and Pszona give a similar
 algorithm for external-memory~\cite{Goodrich11}. Concurrent with our work, 
 Besta et al.~\cite{Besta2020} present a parallel algorithm for generating 
 an $O(\alpha)$-orientation in $O(m)$ work and $O(\log^2n)$ span, 
 which they use for parallel graph coloring.

\myparagraph{Vertex Peeling and $k$-clique Densest Subgraph} An
important application of \kc counting is its use as a subroutine
in computing generalizations of approximate densest subgraph.
In this paper, we study parallel algorithms for \kcds, a
generalization of the densest subgraph problem
introduced by Tsourakakis~\cite{Tsourakakis15}.  Tsourakakis presents
a sequential $1/k$-approximation algorithm based on iteratively
peeling the vertex with minimum \kc-count, and a parallel
$1/(k(1+\epsilon))$-approximation algorithm based on a parallel
densest subgraph
algorithm of  Bahmani et al.~\cite{Bahmani12}. 
Sun \textit{et al.}~\cite{SuDaChSo20} give additional approximation algorithms that converge to produce 
the exact solution over further iterations; these algorithms are more sophisticated and demonstrate 
the tradeoff between running times and relative errors.
Recently, Fang \textit{et al.}~\cite{FaYuChLaLi19} propose algorithms for finding
the largest $(j,\Psi)$-core of a graph, or the largest subgraph
such that all vertices have at least $j$ subgraphs $\Psi$ incident on
them. 
They propose an algorithm for $\Psi$ being a \kc
 that peels vertices with larger clique counts first and show
that their algorithm gives a $1/k$-approximation to the \kcds.

\section{Conclusion}
We presented new work-efficient parallel algorithms for \kc counting and peeling
with low span.  We showed
that our implementations achieve good parallel speedups
and significantly outperform state-of-the-art.
A direction for future work is designing
work-efficient parallel algorithms for the more general $(r,s)$-nucleus
decomposition problem~\cite{Sariyuce2017}.

\section*{Acknowledgments}

This research was supported by NSF Graduate Research Fellowship
\#1122374, 
DOE Early Career Award \#DE-SC0018947,
NSF CAREER Award \#CCF-1845763, Google Faculty Research Award, Google Research Scholar Award, DARPA
SDH Award \#HR0011-18-3-0007, and Applications Driving Architectures
(ADA) Research Center, a JUMP Center co-sponsored by SRC and DARPA.

\bibliographystyle{abbrv}
\bibliography{references}

\appendix

\section{Examples}

\begin{figure*}[!th]
\centering
    \includegraphics[width=\textwidth, page=1]{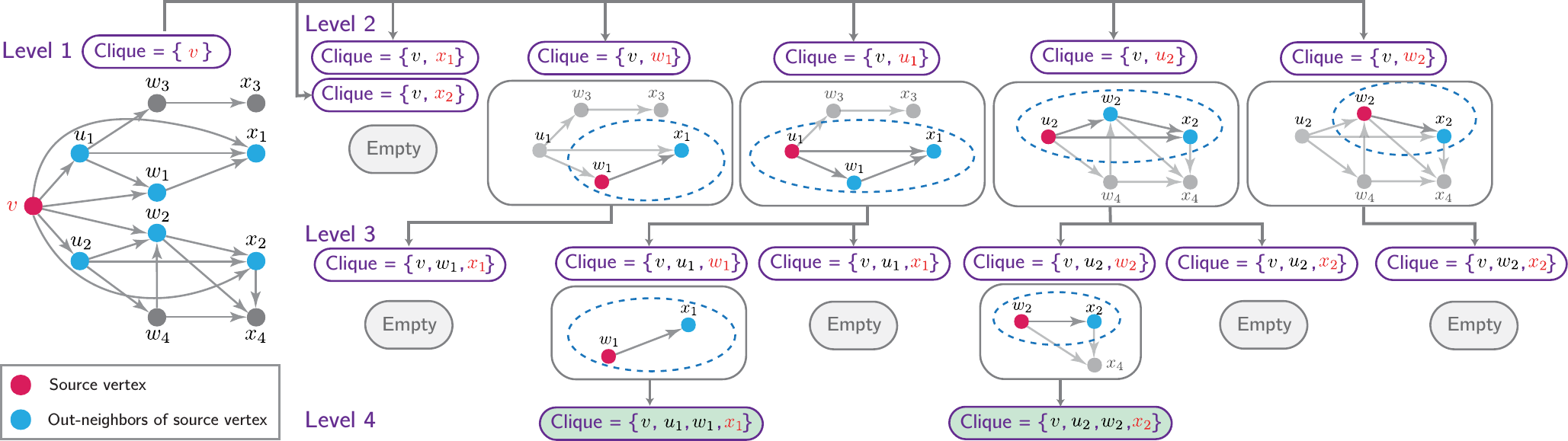}
    \caption{An example of our \kc counting algorithm for $k=4$.}\label{fig:count_example}
\end{figure*}

Figure \ref{fig:count_example} shows an example of our $k$-clique
counting algorithm, \ourcount{}, for $k=4$.  First, the graph is
directed as shown in Level 1. The algorithm then iterates over all
vertices in parallel, but for simplicity we only show a single process
starting from vertex $v$. Importantly, the algorithm must iterate over
all vertices, since there are 4-cliques that are not found by starting
at $v$. For instance, $\{w_2, x_2, w_4, x_4\}$ is a 4-clique that is
found by running this process starting from $w_4$.

We call $v$ the source vertex for its process in Level 1, and note
that $v$ is added to the growing clique set. Each directed
out-neighbor of $v$, in blue, spawns a new child task where the
out-neighbor is added to the growing clique set and is denoted as the
new source vertex, as shown in Level 2.

In Level 2, we take each source vertex, in red, and intersect its
out-neighbors, in blue, with its parent task's out-neighbors,
contained within the dashed blue circle.  Note that all vertices $x_i$
have no out-neighbors, so these tasks terminate here. For simplicity,
we have removed the parent task's source vertex $v$, because no level
2 source vertices may have $v$ as an out-neighbor by virtue of our
orientation, and we have also removed any vertices disconnected from
the source vertex. Now, each vertex in the aforementioned intersection
(i.e., each blue vertex in the dashed blue circle)
spawns a new child task, where it is added to the growing clique set
and becomes the new source vertex, as shown in Level 3. The dashed
blue circle shrinks to contain only the out-neighbors in the
intersection from Level 2.

We repeat this process for a final level, intersecting the
out-neighbors, in blue, of the source vertices, in red, with the
intersection from the previous level, in the dashed blue circle.
Each
vertex remaining in the intersection spawns a new child task. The
child tasks remaining in Level 4 represent our 4-cliques. For example,
the 4-clique $\{v, u_2, w_2, x_2\}$ is obtained by intersecting
$w_2$'s out-neighbors $\{x_2,x_4\}$ and $u_2$'s out-neighbors
$\{w_2,w_4,x_2\}$, which gives $\{x_2\}$ and adding the source
vertices on the path to this task, $\{v, u_2, w_2\}$.  Note that the
arrows between the levels in this figure represent the dependency
graph for the 4-clique counting computation; the spawned child tasks
are all safe to run in parallel.

\section{Proofs}
We present here proofs for the complexity bounds for our low out-degree orientation algorithms, approximate $k$-clique counting algorithm using colorful sparsification, \kc vertex peeling algorithm, and approximate \kc vertex peeling algorithm. Additionally, we prove the variance of our estimator in our approximate \kc counting algorithm, and we prove that the problem of obtaining the $k$-clique cores given by the $k$-clique vertex peeling process is \pcomplete{}.

\subsection{Low Out-degree Orientation (Ranking)}

\begin{algorithm}[!t]
  \footnotesize

 \begin{algorithmic}[1]
   \Procedure{Orient}{$G=(V,E)$, $\epsilon$}
   \State $n \gets |V|, L \gets []$
   \While{$G$ is not empty}
   \State $S \gets \epsilon n/(2+\epsilon)$ vertices of lowest
   induced degree \label{line:gp-peel}
   \State Append $S$ to $L$
   \State Remove vertices in $S$ from $G$
   \EndWhile
   \State \Return $L$
  \EndProcedure
 \end{algorithmic}
\caption{Goodrich-Pszona Orientation Algorithm} \label{alg:goodrich}
\end{algorithm}

\begin{algorithm}[!t]
  \footnotesize
 \begin{algorithmic}[1]
   \Procedure{Orient}{$G=(V,E)$, $\epsilon$, $\alpha$}
   \State $n \gets |V|, L \gets []$
   \While{$G$ is not empty}
   \State $S \gets \{v \in V\ |\ $ $v$'s induced degree less than $
   (2+\epsilon)\alpha\}$ \label{line:be-peel}
   \State Append $S$ to $L$
   \State Remove vertices in $S$ from $G$
   \EndWhile
   \State \Return $L$
  \EndProcedure
 \end{algorithmic}
\caption{Barenboim-Elkin Orientation Algorithm}
 \label{alg:barenboim}
\end{algorithm}

\algprefix{}s \ref{alg:goodrich} and \ref{alg:barenboim} shows pseudocode for the Goodrich-Pszona
algorithm and the
Barenboim-Elkin algorithm respectively. We present here the bounds for our parallelization of
these algorithms, which are described in Section \ref{sec:counting:ranking}.

\goodrichefficient*

\begin{proof}
The bounds on out-degree follow from the discussion in Section~\ref{sec:counting:ranking}. 
Also, as discussed in Section~\ref{sec:counting:ranking}, both algorithms run in
$\BigO{\log n}$ rounds, because each round removes a constant fraction of
the vertices. It remains to prove the work and span bounds for both algorithms.
Note that for both algorithms, we maintain the induced degrees of
all vertices in an array.

For the Goodrich-Pszona algorithm, we can filter out the vertices with
degree less than the $c$'th smallest degree vertex for $c = \epsilon
n/(2+\epsilon)$ using parallel integer sort, which runs in $\BigO{n'}$
work \whp{}, $\BigO{\log n}$ span \whp{}, and $\BigO{n'}$ space,
where $n'$ is the number of remaining vertices~\cite{RaRe89}.  For the
Barenboim-Elkin algorithm, we use a parallel filter, which takes
linear work, $\BigO{\log n}$ span, and linear space.
Overall, the total work to obtain vertices to process in each round for 
both algorithms is $\BigO{n}$ (\whp{} for Goodrich-Pszona), because each round removes
a constant fraction of vertices.

We can update the degrees of the remaining vertices after removing the
peeled vertices by mapping over all edges incident to these vertices,
and applying an atomic add instruction to decrement the degree of each
neighbor. Each edge is processed exactly once in each direction, when
its corresponding endpoints are peeled, and each vertex is peeled
exactly once, so the total work is $\BigO{m}$.  Since each peeling
round can be implemented in $\BigO{\log n}$ span, and there are
$\BigO{\log n}$ such rounds, the span of both algorithms is
$\BigO{\log^2 n}$.

Finally, computing an estimate of the arboricity using the parallel
densest-subgraph algorithm from~\cite{DhBlSh18} can be done in
$\BigO{m}$ work, $\BigO{\log^2 n}$ span, and $\BigO{m}$ space, which
does not asymptotically increase the cost of running the
Barenboim-Elkin algorithm.
\end{proof}

\subsection{Sampling}\label{sec:proofs}

We present here the proof for the variance of our estimator in approximate $k$-clique counting through colorful sparsification, which is described in Section~\ref{sec:counting:sampling}.

\begin{theorem}\label{thm:unbiased}
Let $X$ be the true \kc count in $G$, $C$ be the \kc count in $G'$,
$p=1/c$, and $Y=C/p^{k-1}$. Then $\mathbb{E}[Y] = X$ and
$\mathbb{V}ar[Y] = p^{-2(k-1)}(X(p^{k-1}-p^{2(k-1)}) +
\sum_{z=2}^{k-1}s_z(p^{2(k-1)-z+1}-p^{2(k-1)}))$, where $s_z$ is the
number of pairs of \kc{}s that share $z$ vertices.
\end{theorem}

\begin{proof}
Let $C_i$ be an indicator variable denoting whether the $i$'th \kc in
$G$ is preserved in $G'$. For a \kc to be preserved, all $k$ vertices
in the clique must have the same color. This happens with probability
$p^{k-1}$ since after fixing the color of one vertex $v$ in the
clique, the remaining $k-1$ vertices must have the same color as
$v$. Each vertex picks a color independently and uniformly at random,
and so the probability of a vertex choosing the same color as $v$ is
$p$.  The number of \kc{}s in $G'$ is equal to
$C=\sum_i{C_i}$. Therefore $\mathbb{E}[C] = \sum_i \mathbb{E}[C_i] =
Xp^{k-1}$. We have that $\mathbb{E}[Y] = \mathbb{E}[C/p^{k-1}] =
(1/p^{k-1})\mathbb{E}[C] = (1/p^{k-1})Xp^{k-1}=X$.

The variance of $Y$ is $\mathbb{V}ar[Y] = \mathbb{V}ar[C/p^{k-1}] =
\mathbb{V}ar[C]/p^{2(k-1)}=\mathbb{V}ar[\sum_iC_i]/p^{2(k-1)}$. By
definition $\mathbb{V}ar[\sum_iC_i] = \sum_i
(\mathbb{E}[C_i]-\mathbb{E}[C_i]^2) + \sum_{i\neq
  j}\mathbb{C}ov[C_i,C_j]$. The first term is equal to
$X(p^{k-1}-p^{2(k-1)})$. To compute the second term, note that
$\mathbb{C}ov[C_i,C_j]=\mathbb{E}[C_iC_j]-\mathbb{E}[C_i]\mathbb{E}[C_j]$
depends on the number of vertices that cliques $i$ and $j$
share. Their covariance is $0$ if they share no vertices.  Their
covariance is also $0$ if they share one vertex since the event that
the remaining vertices of each clique have the same color as the
shared vertex is independent between the two cliques.  Let $s_z$
denote the number of pairs of cliques that share $z>1$ vertices. For
pairs of cliques sharing $z$ vertices, we have
$\mathbb{E}[C_iC_j]=p^{2(k-1)-z+1}$. This is because after fixing the color of one of the shared vertices, there are $2(k-1)-z+1$  remaining vertices in the pair of cliques to color, and the probability that they all match the fixed color is $p^{2(k-1)-z+1}$.
Therefore $\mathbb{C}ov[C_i,C_j]=\mathbb{E}[C_iC_j]-\mathbb{E}[C_i]\mathbb{E}[C_j] =
p^{2(k-1)-z+1}-p^{2(k-1)}$.
In total, we have $\mathbb{V}ar[Y] =
p^{-2(k-1)}(X(p^{k-1}-p^{2(k-1)}) +
\sum_{z=2}^{k-1}s_z(p^{2(k-1)-z+1}-p^{2(k-1)}))$.
\end{proof}

We also give the proof for the work and span of approximate $k$-clique
counting through colorful sparsification.

\sampworkspan*

\begin{proof}
Theorem~\ref{thm:unbiased} states that our algorithm produces an
unbiased estimate of the global count. We now analyze the work and
span of the algorithm. Choosing colors for the vertices can be done in
$O(n)$ work and $O(1)$ span. Creating a subgraph containing edges with
endpoints having the same color can be done using prefix sum and
filtering in $O(m)$ work and $O(\log m)$ span.  Each edge is kept with
probability $p$ as the two endpoints will have matching colors with
this probability. Therefore, our subgraph has $pm$ edges in
expectation. The arboricity of our subgraph is upper bounded by the
arboricity of our original graph, $\alpha$, and so including the work
of running our \kc counting algorithm on the subgraph, the sampling
algorithm takes $O(pm\alpha^{k-2}+m)$ work and $O(k \log n + \log^2
n)$ span \whp{}.
\end{proof}

\subsection{Vertex Peeling}

We present here the proof that \ourpeel correctly generates a subgraph
with the same approximation guarantees of Tsourakakis' sequential
\kcds algorithm~\cite{Tsourakakis15}, as well as the following bounds on the
complexity of \ourpeel.

\peelexact*

\begin{proof}

This proof uses the Nash-Williams theorem~\cite{nash1961edge}, which
states that a graph $G$ has arboricity $\alpha$ if and only if for
every $U \subseteq V$, $|G[U]| \leq \alpha(|U| - 1)$.  Here, $G[U]$ is
the subgraph of $G$ induced by the vertices in $U$, and $|G[U]|$ is
the number of edges in $G[U]$.

First, we provide a proof on the correctness of our \kc peeling
algorithm. Tsourakakis~\cite{Tsourakakis15} proves that the sequential
$k$-clique densest subgraph algorithm that peels vertices one by one
in increasing order of \kc count attains a $1/k$-approximation to the
$k$-clique densest subgraph problem. We note that among vertices
within the same \kc core, the order in which these vertices are peeled
does not affect the approximation; this follows directly from
Tsourakakis's sequential algorithm, in which in any given round, any
vertex with the same minimum \kc count may be peeled. Additionally,
given a set of vertices with the same minimum \kc count in any given
round, peeling a vertex from this set does not change the \kc core
number of any other vertex in this set by definition.

As such, in order to show the correctness of \ourpeel, it suffices to
show that first, \ourpeel peels vertices in the same order as given by
Tsourakakis's sequential algorithm, except for vertices in the same
\kc core which may be peeled in any order among each other, and
second, \ourpeel correctly updates the \kc counts after peeling these
vertices. The first claim follows from the structure of \ourpeel,
because \ourpeel peels all vertices with the same minimum \kc count in
each round, which may be serialized to any order. The second claim
follows from the correctness of \ourcount in
Section~\ref{sec:counting}.  \ourpeel first obtains for each peeled
vertex $v$ all undirected neighbors $N(v)$, and then uses the
subroutine from \ourcount to obtain $k$-cliques originating from
$v$. This is equivalent to first obtaining the induced subgraph on
$N(v)$, and then performing \ourcount to obtain $(k-1)$-cliques on the
induced subgraph; thus, this gives all $k$-cliques containing $v$. The
additional filtering of already peeled vertices ensures that
previously found $k$-cliques are not recounted, and filtering other
vertices peeled simultaneously based on a total order ensures that
$k$-cliques involving these vertices are not recounted.  Thus,
\ourpeel gives a $1/k$-approximation to the $k$-clique densest
subgraph problem.

The proof of the complexity is similar in spirit to that of
Theorem~\ref{thm:clique-list}, but there are some subtle and important
differences. First, unlike in our \kc counting algorithm
(\algprefix~\ref{alg-count}), our peeling algorithm does not have the
luxury of only finding \kc{}s directed out of a peeled vertex
$v$. Instead, it must find \emph{all} \kc{}s that $v$ participates in
and decrement these counts.  Arguing that this does not
cost a prohibitive amount of work is the main challenge of the
proof. Importantly, our peeling algorithm calls the recursive
subroutine \algname{rec-count-cliques-v} of our \kc counting algorithm
directly, on a different input than used in the full \kc counting
algorithm, and so the analysis of the work and span differs from the
analysis given in Section~\ref{sec:counting:listing}.

We first account for the work and span of extracting and updating the
bucketing structure. The overall work of inserting vertices into the
structure is $O(n)$. Each vertex has its bucket
decremented at most once per \kc, and since there are at most
$O(m\alpha^{k-2})$ $k$-cliques, this is also the total cost for updating buckets
of vertices. Lastly, removing the minimum bucket can
be done in $O(\log n)$ amortized expected work and $O(\log n)$ span
\whp{}, which costs a total of $O(\rho_k(G)\log n)$ amortized expected
work, and $O(\rho_k(G)\log n)$ span \whp{}.

Next, to bound the cost of finding all \kc{}s incident to a peeled
vertex $v$, we rely on the Nash-Williams theorem, which provides a
bound on the size of induced subgraphs in a graph with arboricity
$\alpha$.  Notably, in the first level of recursion when \algname{rec-count-cliques-v} 
is called from the \algname{Update} subroutine,
the intersect operations performed essentially compute the induced
subgraph on the neighbors of each peeled vertex $v$; this is because
during this call, we intersect the directed neighbors of each
vertex in $N_G(v)$ (that has not been previously peeled) with $N_G(v)$
itself, producing a pruned version of the induced subgraph of $N_G(v)$
on $G$.  We have that for each $v \in V$, the induced subgraph on its
neighbors has size $|G[N(v)]| \leq \alpha(|N(v)| - 1) = \alpha(d(v) -
1)$. Assuming for now that we can construct the induced subgraph on
all vertex neighborhoods in work linear in their size, summed over all
vertices, the overall cost is just
  \begin{equation}\label{eqn:nwhood}
    \sum_{v \in V} |G[N(v)]| \leq \alpha \sum_{v \in V} (d(v) - 1) =
    \BigO{m\alpha}
  \end{equation}

How do we build these subgraphs in the required work and span? Our
approach is to do so using an argument similar to the elegant proof
technique proposed in Chiba-Nishizeki's original $k$-clique listing
algorithm~\cite{ChNi85}. Because the first call to
\algname{rec-count-cliques-v} takes each vertex $u \in N_G(v)$ and
intersects the directed neighbors $N_{DG}(u)$ with $N_G(v)$, we use
$\BigO{\min(d(u), d(v))}$ work to build the induced subgraph on $v$'s
neighborhood. Observe that each edge in the graph is processed by an
intersection in this way exactly once in each direction, when
each endpoint is peeled.  By Lemma 2 of~\cite{ChNi85}, we know that $
\sum_{e=(u,v) \in E} \min(d(u), d(v)) = \BigO{m\alpha}$ and therefore
the overall work of performing all intersections is bounded by
$\BigO{m\alpha}$, and the per-vertex induced subgraphs can therefore
also be built in the same bound. The span for this step is
$\BigO{\log n}$ \whp{}~using parallel hash tables~\cite{gil91a}.

Lastly, we account for the remaining cost of performing $k$-clique
counting within each round. We now recursively call
\algname{rec-count-cliques-v} $k-1$ times in total, as $\ell$ ranges
from $1$ to $k-1$, but the final recursive call returns the size of
$I$ immediately, and we have already discussed the work of the first
call to \algname{rec-counts-v}. Considering the remaining $k-3$
recursive steps with non-trivial work, we have $\BigO{m'\alpha^{k-3}}$
work and $\BigO{k \log n}$ span where $m'$ is the size of the
vertex's induced neighborhood.  Considering the work first, summed
over all vertices' induced neighborhoods, the total work is
  \[
    \sum_{v \in V} \operatorname{O}\left( |G[N(v)]| \alpha^{k-3} \right) =
    \BigO{m\alpha^{k-2}}
  \]
which follows from Equation~\ref{eqn:nwhood}. The span follows, since
adding in the span of the first recursive call, we have
$\BigO{k \log n}$ span to update $k$-clique counts per peeled
vertex, and there are $\rho_k(G)$ rounds by definition.
\end{proof}

\subsection{Parallel Complexity of Vertex Peeling}\label{apx:complexity}

The $k$-clique vertex peeling algorithm exactly computes for all $c$ the
\emph{\kc $c$-cores} of the graph. The \defn{\kc $c$-core} of a graph
$G$ is defined as the maximal subgraph such that every vertex is
contained within at least $c$ \kc{}s.
This is a generalization of the classic \defn{$c$-core} problem, which
is the maximal subgraph such that each vertex has degree at least $c$.
The $c$-core problem is well known to be \pcomplete{} for $c \geq
3$~\cite{anderson84pcomplete}.
We present here the proof that computing \kc $c$-cores is
\pcomplete{} for $c>2$. We first observe that since the
number of $k$-cliques incident to each vertex can be efficiently
computed in \nc{} by
Theorem~\ref{thm:clique-list}, the problem of computing \kc $1$-cores is
in \nc{} for constant $k$.

\myparagraph{\kc $\bm{c}$-cores when $\bm{c > 2}$}
Next, we study the parallel complexity of computing \kc $c$-cores
for $c > 2$. We will show that there is an \nc{} reduction from the
problem of deciding whether the $c$-core is non-empty, to the problem of
deciding whether the \kc $c$-core is non-empty.
We first discuss the reduction at a high level. The input is a graph
and some value $c$, and the problem is to decide whether the $c$-core
is non-empty. The idea is to map the original peeling process to
compute the $c$-core to a peeling process to compute the \kc $c$-core.
Figure~\ref{fig:one_s_reduction} shows the reduction for $k=3, c=7$,
and marks the initial \kc degrees of each vertex in red.

\begin{figure}[t!]
  \centering
    \includegraphics[width=0.45\textwidth]{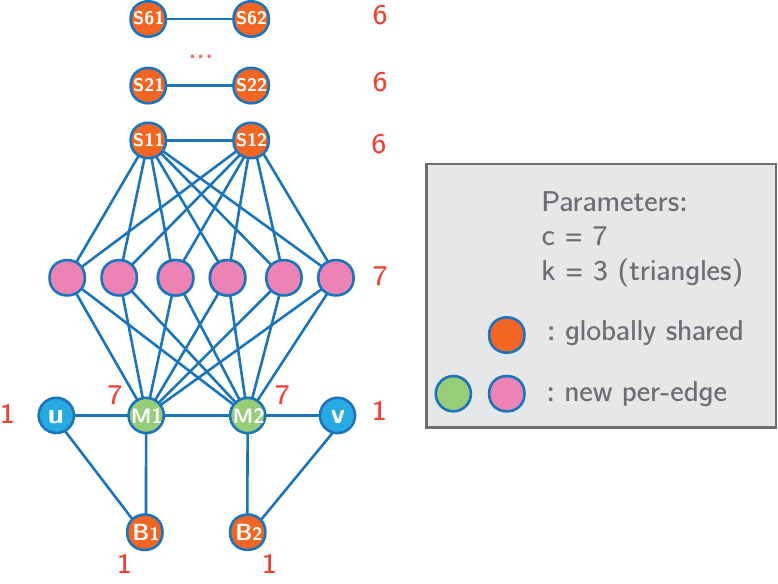}
    \caption{Gadget used for the reduction for computing 3-clique
    $c$-cores for an original edge $(u,v)$, shown for $c=7$. The red
    numbers next to each vertex show the vertex's $3$-clique degree, or the
    number of $3$-cliques incident per vertex before peeling (as a
    contribution of the gadget for this $(u,v)$ edge). The pink gadget
    vertices and green middle vertices are distinct per edge in the
    reduction. The orange special vertices and base vertices are
    created only once and globally shared.
    The gadget shown here gives each of the green middle vertices,
    $M_1$ and $M_2$, a $3$-clique degree of $c=7$, with one of the incident
    triangles being formed by an original vertex endpoint (either $u$
    or $v$), and the remaining being formed with the pink gadget
    vertices. Each of the pink gadget vertices forms one triangle with
    the green middle vertices, and the remaining $c-1$ triangles with
    $c-1$ pairs of special (globally shared) vertices, $S_{ij}$.
    Modifying this construction for different values of $c$ is done by creating
    $c-1$ pink gadget vertices and $c-1$ special-vertex pairs.
    To generalize to $k > 3$ we make the base vertices $B_{1}$
    and $B_{2}$ each a $(k-2)$-clique which are fully connected to
    $u$ and one of the green middle vertices to form \kc{}s.
    Similarly, each pink gadget vertex becomes a $(k-2)$-clique to form
    $c-1$ \kc{}s with the green middle vertices.
    }\label{fig:one_s_reduction}
\end{figure}

The reduction works as follows. We break up each edge $(u,v)$ in the
graph into four vertices connected in a path with the left-most and
right-most vertices corresponding to the original vertices $u$ and
$v$. The middle vertices are $M_{1}$ and $M_{2}$ (shown in green in
Figure~\ref{fig:one_s_reduction}). We create gadgets to increase each
of the two middle vertices' \kc{} degrees to $c$. The gadgets are
constructed so that if either of the original vertices
has its \kc{} degree go below $c$ and is thus not in the \kc $c$-core,
then the path corresponding to this edge will unravel, and the other
original vertex will have its \kc{} degree decremented by one, exactly
as in the $c$-core peeling process.

The reduction constructs $c-1$ $k$-cliques between the two middle
vertices in the path, and a set of new gadget vertices for this edge
(shown in pink in Figure~\ref{fig:one_s_reduction}). It also
constructs one \kc between each original edge endpoint, its
neighboring middle vertex, and a specially designated set of base
vertices, which are globally shared. The last part of the construction
ensures that the gadget vertices have large enough \kc degrees by
creating $c-1$ $k$-cliques between them and a set of special vertices,
$S_i$'s, which are globally shared.

To argue that this reduction is correct, it suffices to show that the
\kc $c$-core is non-empty if and only if the $c$-core of the original
graph is non-empty. We only argue the reverse direction, since the
proof for the forward direction is almost identical. Suppose the input
graph has a non-empty $c$-core, $C$.  Then, observe that all of the
original vertices corresponding to $C$ in the reduction graph
have \kc degree at least $c$.  Furthermore, since all of the middle
vertices that are added for original edges in the $c$-core initially
have \kc degree exactly $c$, the gadget vertices corresponding to
these edges have \kc degree exactly $c$. It remains to argue that the
special vertices and the base vertices have sufficient \kc degree.
Observe that the special vertices connected to all gadget vertices for
the edges form $k$-cliques with all gadget vertices, and thus have \kc
degree $(c-1)|E(G[C])|$, where $E(G[C])$ is the set of edges in the
induced subgraph on $C$. Since a $c$-core on $|C|$ vertices must have
at least $c|C|$ edges, $(c-1)|E(G[C])| \geq (c-1)c|C| > c$.  Similarly
for the base vertices, they form a single \kc for each edge in the
$c$-core, and so the base vertices have \kc degree $c|C|$.
Thus, the subgraph corresponding to the original vertices, middle
vertices, and gadget vertices for these edges, and the special and
base vertices all have sufficient \kc degree to form a non-empty \kc
$c$-core.

By the discussion above, we have shown that computing \kc $c$-cores is
\pcomplete{} for $c>2$, a strengthening of the original
\pcomplete{}ness result of Anderson and
Mayr~\cite{anderson84pcomplete}.

An interesting question is to understand the parallel complexity of
computing \kc $2$-cores for any constant $k$. For $k=2$, Anderson
and Mayr observed that this problem is in \nc{}.
We leave it for future work to
determine whether a similar algorithm can find the \kc $2$-cores
in \nc{} for $k > 2$.

\subsection{Approximate Vertex Peeling}

We now present the proof for the work and span of our approximate $k$-clique peeling algorithm, \ourapproxpeel, from Section \ref{sec:peeling:approx}.

\approxpeeling*
\begin{proof}
The correctness and approximation guarantees of this algorithm follows
from~\cite{Tsourakakis15}. The work bound follows similarly to the
proof of Theorem~\ref{thm:peelexact}. Also, filtering the
remaining vertices in each round can be done in $O(n)$ total
work, since a constant fraction of vertices are peeled in
each round. The span bound follows because there are
$O(\log n)$ rounds in total, and since each round runs in
$\BigO{k \log n}$ span, again using the same argument given in Theorem \ref{thm:peelexact}.
\end{proof}

\begin{table*}[t]
\small
\centering
\setlength{\tabcolsep}{2pt}

\scalebox{0.9}{
\begin{tabular}{lllllllll}
\toprule
 & $k=4$      & $k=5$      & $k=6$       & $k=7$       & $k=8$       & $k=9$       & $k=10$      & $k=11$               \\ \midrule
\textbf{com-dblp} &16,713,192 & 262,663,639 & $4.222 \times 10^9$ & $6.091 \times 10^{10}$ & $7.772\times 10^{11}$ & $8.813 \times 10^{12}$ & $8.956 \times 10^{13}$& --- \\ \hline
\textbf{com-orkut}& $3.222 \times 10^9$ & $1.577 \times 10^{10}$ &  $7.525 \times 10^{10}$& $3.540 \times 10^{11}$ & $1.633 \times 10^{12}$ & $7.248 \times 10^{12}$ & $3.029\times 10^{13}$ & $1.171 \times 10^{14}$\\ \hline
\textbf{com-friendster}&  $8.964 \times 10^9$ & $2.171 \times 10^{10}$ & $5.993 \times 10^{10}$ & $2.969 \times 10^{11}$ & $3.120 \times 10^{12}$ & $4.003 \times 10^{13}$ & $4.871 \times 10^{14}$& --- \\ \hline
\textbf{com-lj}& $5.217 \times 10^9$ & $2.464 \times 10^{11}$ & $1.099 \times 10^{13}$ & $4.490 \times 10^{14}$& ---&---&---&--- \\ \hline
\textbf{ClueWeb}&  $2.968\times 10^{14}$ & --- & --- & ---& ---&---&---&---\\ \hline
\textbf{Hyperlink2014}& $7.500\times 10^{14}$ & --- & --- & ---& ---&---&---&---\\ \hline
\textbf{Hyperlink2012}&   $7.306\times 10^{15}$ & --- & --- & ---& ---&---&---&---\\ \hline
\end{tabular}
}
\caption{Total $k$-clique counts for our input graphs. Note that we do not have statistics for certain graphs for large values of $k$, because  algorithm did not terminate in under 5 hours; these entries are represented by a dash.}\label{table:graph-stats2}
\end{table*}

\section{Additional Experiments and Data}\label{sec:additional-figures}
This section presents additional experimental data from our evaluation.

\subsection{Counting Results}

Figure~\ref{table:graph-stats2} shows the \kc counts that we obtained from our algorithms. Figure \ref{fig:distribution} shows the frequencies of the 4-clique counts per vertex, as obtained using \ourcount{}, on the large ClueWeb and Hyperlink2014 graphs. We see that the number of vertices decreases roughly exponentially as a function of the 4-clique count.

\begin{figure*}[!t]
    \centering
   \begin{subfigure}{.49\textwidth}
   \centering
   \includegraphics[width=.8\columnwidth, page=3]{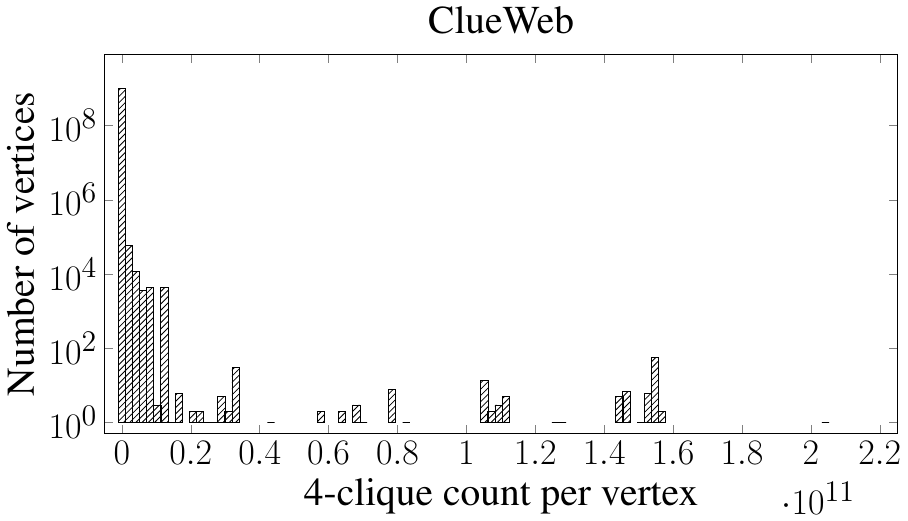}
   \end{subfigure}%
   \begin{subfigure}{.49\textwidth}
   \centering
   \includegraphics[width=.8\columnwidth, page=2]{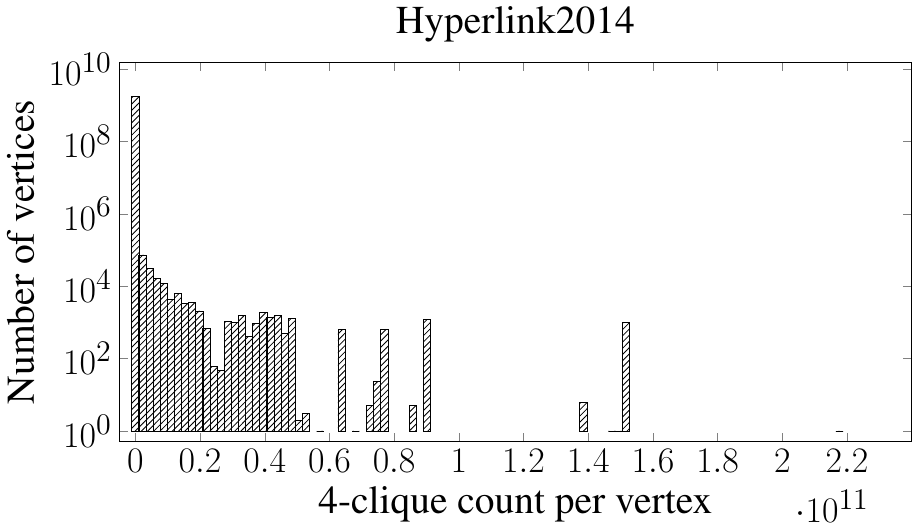}
      \end{subfigure}
    \caption{The frequencies of 4-clique counts per vertex, obtained using \ourcount, on the large graphs ClueWeb and Hyperlink2014. Note that the frequencies is plotted on a log-scale.}\label{fig:distribution} 
\end{figure*}

Figure~\ref{fig:ranks} shows the
preprocessing overheads and \kc counting times for
com-orkut using different orientations and node parallelism. 
The Goodrich-Pszona orientation is 2.86x slower than the
orientation using degree ordering, but this overhead is not
significant for large $k$ and the Goodrich-Pszona
orientation gives the fastest counting times for large
$k$. We found that the self-relative speedups of orienting the
graph alone were between 6.69--19.82x across all orientations, the
larger of which were found in the larger graphs.

\begin{figure*}[!t]
  \centering
\includegraphics[width=0.8\textwidth, page=2]{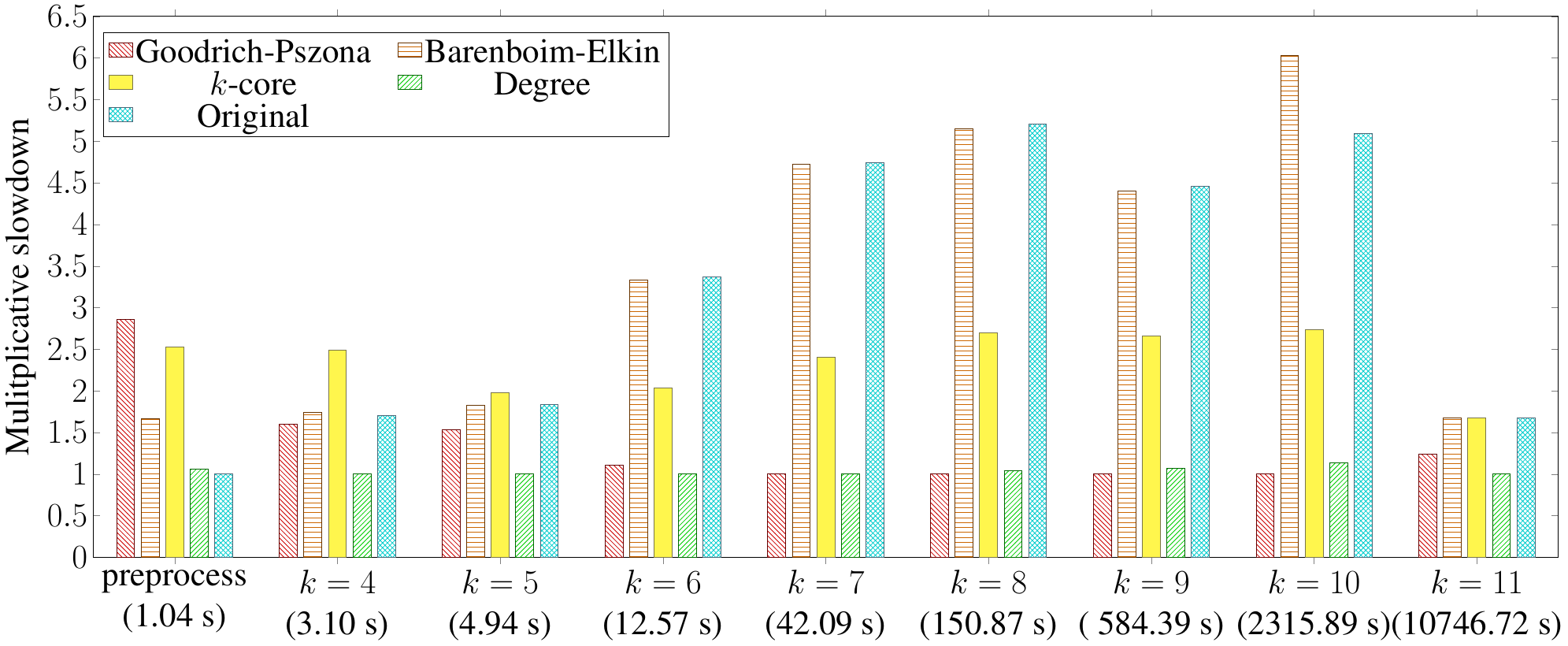}
  \caption{Parallel runtimes for \kc counting (\ourcount{}) on
  com-orkut for different orientations, using node parallelism. All
  times are scaled by the fastest parallel time (indicated in
  parentheses). The first set of bars show the preprocessing overhead of the different orientations. The remaining sets of bars show the performance including both preprocessing and counting.}\label{fig:ranks} 
\end{figure*}

In both \ourcount{} and \kclist{}, node parallelism is faster on small $k$, while
edge parallelism is faster on large $k$, due to the greater work required on large $k$
and the additional parallelism available in edge parallelism to take advantage
of this work. The
cutoff for when edge parallelism is generally faster than node
parallelism occurs around $k=8$. 
Figure~\ref{fig:node_edge} shows this behavior in \ourcount{}'s $k$-clique counting runtimes
on com-orkut, where for $k\geq 9$ edge parallelism becomes faster than
node parallelism.

Figure~\ref{fig:approx_count_friend} 
show the runtimes for our approximate counting algorithm on com-orkut and com-friendster.
We see that there is an inflection point
where after enough sparsification, obtaining \kc counts for
large $k$ is faster than for small $k$; this is because we cut off the
recursion when there are not enough vertices to
complete a \kc. We compute our error rates as $|\text{exact} -\allowbreak \text{approximate}| / \text{exact}$.
For $p=0.5$ on both com-orkut and com-friendster across
all $k$, we see between 2.42--87.56x speedups over exact counting (considering the best \kc counting runtimes) and between
0.39--1.85\% error. Our error rates degrade for higher $k$ and
lower $p$, but even for $p$ as low as $0.125$, we obtain between
5.32--2189.11x speedups over exact counting and between 0.42--5.05\% error.

We compare our approximate counting algorithms to approximate \kc
counting using \algname{MOTIVO}~\cite{Bressan2019}. We ran both the
naive sampling and adaptive graphlet sampling (AGS) options in \algname{MOTIVO}.  However,
\algname{MOTIVO} is unable to run on com-friendster because it runs
out of memory. Moreover, in order to achieve a 6\% error rate,
\algname{MOTIVO} takes between 92.71--177.29x the time that our
algorithm takes for 4-clique and 5-clique approximate counting on
com-orkut. \algname{MOTIVO} takes 168.84 seconds to approximate
6-clique counts with 31.55\% error, while our algorithm takes 0.49
seconds to approximate 6-clique counts with under 6\% error. However,
unlike our algorithm, \algname{MOTIVO} can estimate non-clique
subgraph counts.

\subsection{Peeling Results}

Figure~\ref{table:graph-peel-stats}
shows the peeling rounds, \kc core sizes, and approximate maximum \kc densities
that we obtained from our algorithms.

Figure~\ref{fig:buckets} shows the frequencies of the different
numbers of vertices peeled in each parallel round for com-orkut,
for $4 \leq k \leq 6$. A significant number of rounds contain
fewer than 50 vertices peeled, and by the time we reach the tail of
the histogram, there are very few parallel rounds with a large number
of vertices peeled.

Figures~\ref{fig:approx_peel_orkut} and~\ref{fig:approx_peel_friend} show the parallel runtimes of \ourapproxpeel and of \kclist's approximate peeling algorithm on com-orkut and com-friendster, respectively.  Note that while \kclist only provides a sequential implementation for exact $k$-clique peeling, \kclist provides a parallel implementation for approximate $k$-clique peeling.
We found
there was not a significant difference in performance across different
values of $\epsilon$ for both implementations. 
\ourapproxpeel is up to 29.59x
faster than \kclist{} for large $k$. However, \ourapproxpeel was
slower on com-friendster for small $k$ since \kclist{} uses a serial
heap to recompute vertices to be peeled; this is more amenable over
rounds containing fewer vertices, while our implementation incurs additional
overhead to recompute peeled vertices in parallel, which is mitigated
over larger $k$.  In terms of
percentage error in the maximum $k$-clique density obtained compared
to the density obtained from $k$-clique peeling, we see between
48.58--77.88\% error on com-orkut and 5.95--80.83\% error on
com-friendster.

\begin{figure*}
  \centering
\includegraphics[width=\textwidth, page=1]{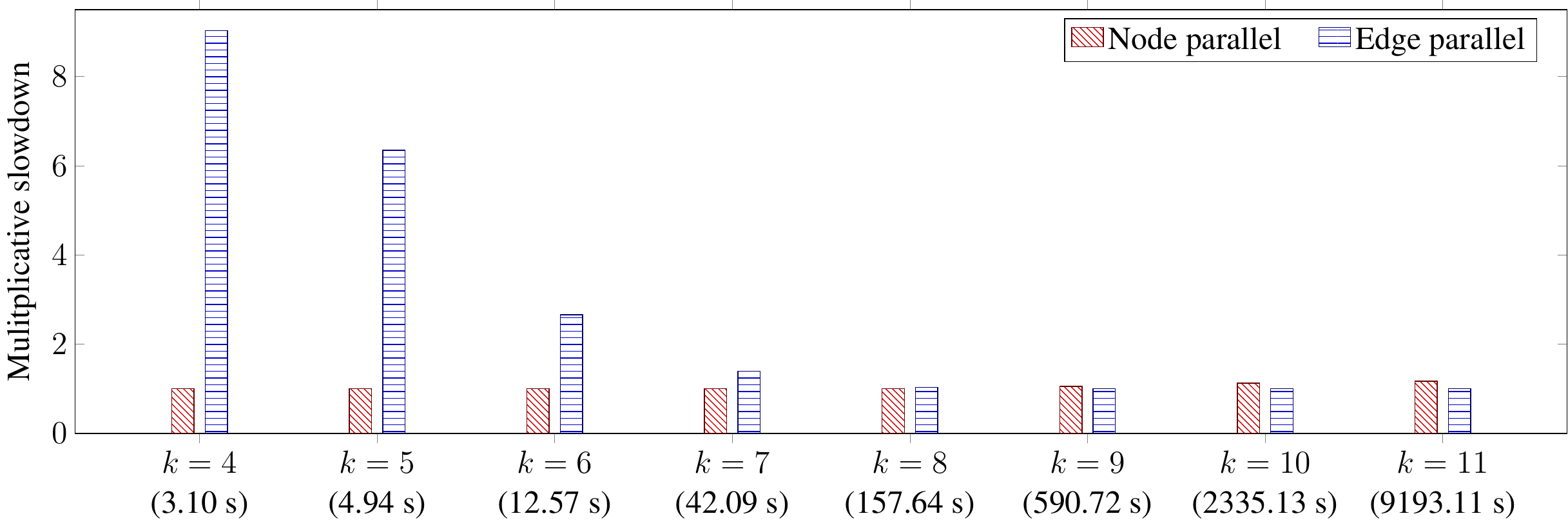}
  \caption{Parallel runtimes for $k$-clique counting (\ourcount{}) on com-orkut, considering node parallelism and edge parallelism, and fixing the orientation given by degree ordering. All times are scaled by the fastest parallel runtime, which are given in the parentheses. }\label{fig:node_edge} 
\end{figure*}

\begin{figure*}[!t]
    \centering
   \begin{subfigure}{.3\textwidth}
   \centering
   \includegraphics[width=\textwidth, page=3]{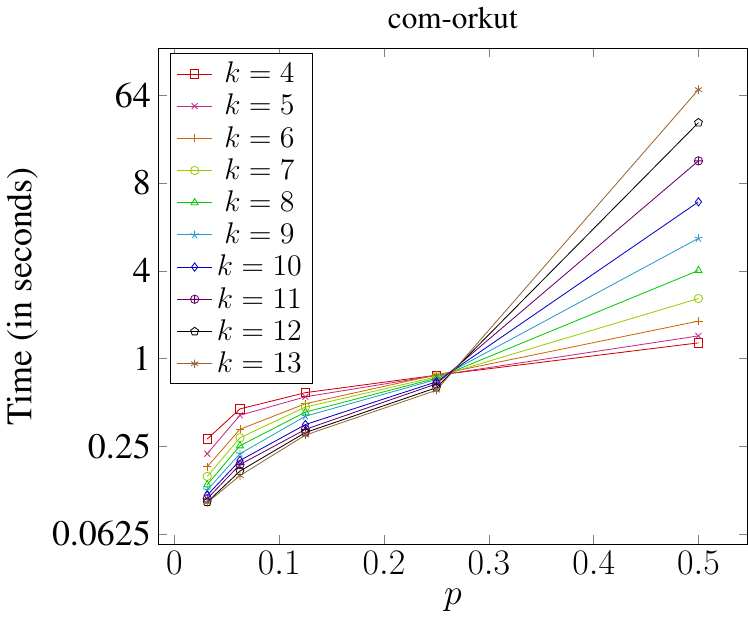}
   \end{subfigure}%
   \hspace{2cm}
   \begin{subfigure}{.3\textwidth}
   \centering
   \includegraphics[width=\textwidth, page=4]{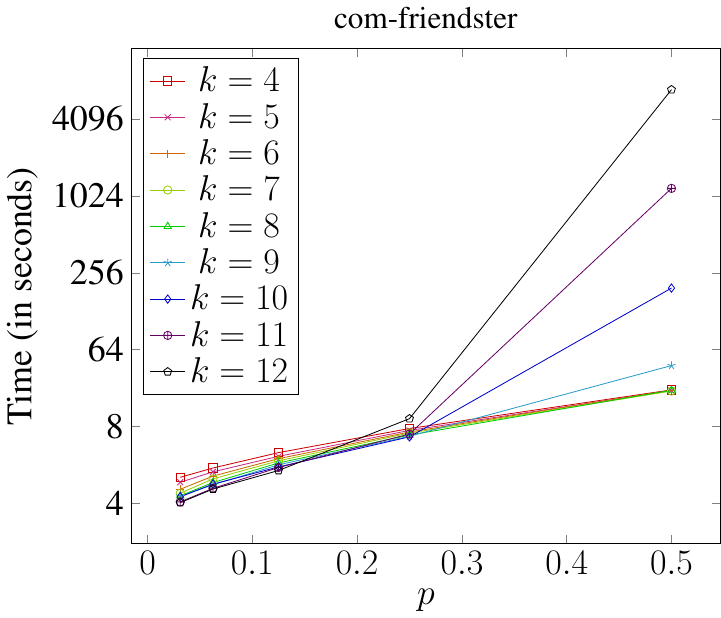}
      \end{subfigure}
    \caption{Parallel runtimes for approximate $k$-clique counting (\ourcount{}) on com-orkut and com-friendster, varying over $p = 1/c$  where $c$ is the number of colors used. The runtimes were obtained using the orientation given by degree ordering and node parallelism. The $y$-axis is in log-scale.}
  \label{fig:approx_count_friend}
\end{figure*}

\begin{table*}[t!]
  \small
\centering
\scalebox{0.9}{
\begin{tabular}{lllllllll}
\toprule
&    &  $k=3$ &$k=4$                & $k=5$   & $k=6$     & $k=7$       & $k=8$      & $k=9$              \\ \midrule
& $\rho_k$          & 5,626    & 11,669   & 17,720     & 22,091       & 23,988      & 23,095       & 21,538       \\
& Max density ($k$-clique peeling) &10,100&297,096 & 5,598,323 & 75,372,336  & 782,071,056	 & $6.183 \times 10^9$ & $4.080 \times 10^{10}$ \\
\hline

& $\rho_k$        &     483     & 441     & 343       & 240         & 166        & 127         & 103         \\
\textbf{com-dblp} & $k$-clique core  &   6,328    & 234,136  & 6,438,740   & 140,364,532   & $2.527\times 10^9$ & $3.862\times 10^{10}$ & $5.117\times 10^{11}$ \\
& Max density ($k$-clique peeling) &6,328 &234,136 &  6,438,740 & 140,364,532 & $2.527\times 10^9$ & $3.862 \times 10^{10}$ & $5.117 \times 10^{11}$ \\
 \hline

& $\rho_k$           &    36,752  & 94,931   & 160,577    & 210,966      & 236,623     & 241,330      &    ---         \\
\textbf{com-orkut} & $k$-clique core  & 7,117     & 117,182  & 2,115,900   & 29,272,988    & 312,629,724  & $2.741\times 10^9$  &    ---         \\ 
& Max density ($k$-clique peeling)&18,547& 340,997 & 4,882,477 & 73,696,814 & 883,634,847 & $8.332 \times 10^9$&  --- \\
\hline

& $\rho_k$         &       57,090    & 140,705  & 249,605    & 339,347      & ---     & ---      &      ---       \\
\textbf{com-friendster}& $k$-clique core   & 8,255   & 349,377  & 11,001,375  & 274,901,025   &        ---    &        ---     &        ---     \\ 
& Max density ($k$-clique peeling) & 9,521& 428,928 & 12,762,919 & 363,676,399&  --- &   ---&  --- \\
\hline

& $\rho_k$           & 13,899  & 29,514   & 42,994     & 50,159       &      ---      &      ---       &    ---         \\
\textbf{com-lj} & $k$-clique core   &  64,478    & 7,660,975 & 679,343,769 & $4.796\times 10^{10}$ &        ---    &      ---       &    ---        \\
& Max density ($k$-clique peeling) &72,255& 9,031,923 & 839,813,448 &$6.199 \times 10^{10}$&  --- &   ---&  --- \\
\end{tabular}
}
\caption{Relevant $k$-clique peeling statistics for the SNAP graphs that we experimented on. We do not have statistics for certain graphs for large values of $k$, because the corresponding $k$-clique peeling algorithms did not terminate in under 5 hours; these entries are represented by a dash. }\label{table:graph-peel-stats}
\end{table*}

\begin{figure*}[t]
\includegraphics[width=1\textwidth, page=7]{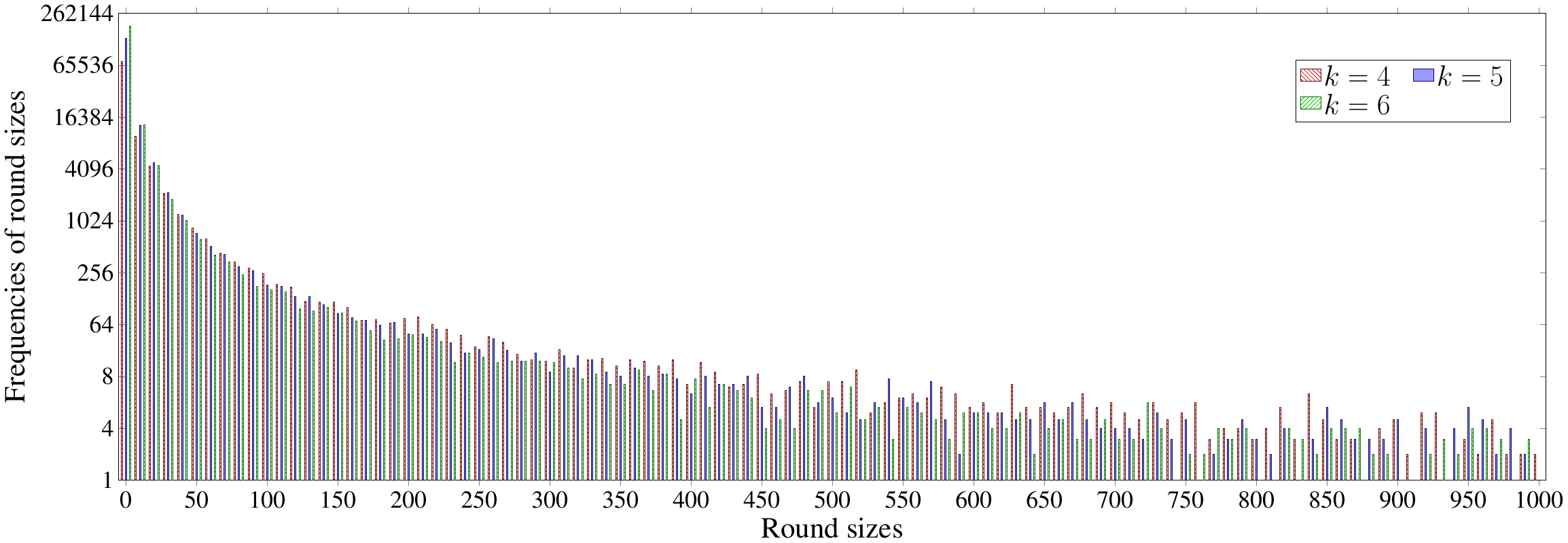}
  \caption{Frequencies of the number of vertices peeled in a parallel round using \ourpeel{}, for $k$-clique peeling on com-orkut ($4 \leq k \leq 6$). Rounds with more than 1000 vertices peeled have been truncated; these truncated round frequencies are very low, most often consisting of 0 rounds. The frequencies are given in log-scale.}\label{fig:buckets} 
\end{figure*}

\begin{figure*}[!t]
\centering
\includegraphics[width=1.1\columnwidth, page=5]{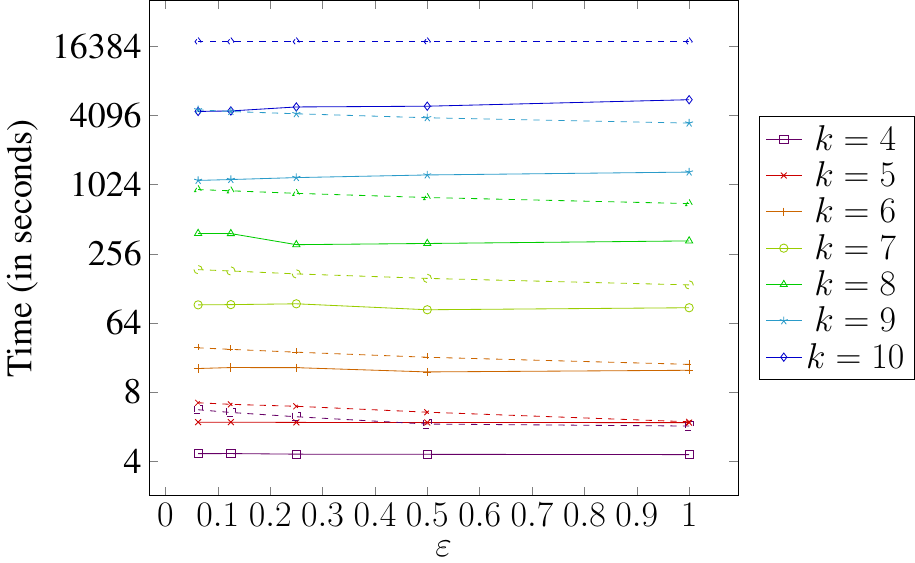}
  \caption{Parallel runtimes on com-orkut for approximate $k$-clique peeling using \ourapproxpeel{} (solid lines) and \kclist{} (dashed lines). These runtimes were obtained by varying over $\epsilon$, giving a $1/(k(1+\epsilon))$-approximation of the $k$-clique densest subgraph. These runtimes were obtained using the orientation given by degree ordering, and the runtimes are given in log-scale. Moreover, we cut off \kclist{}'s runtimes at 5 hours, which occurred for $k=10$ over all $\varepsilon$.  }\label{fig:approx_peel_orkut} 
\end{figure*}

\begin{figure*}[!t]
  \centering
\includegraphics[width=0.7\columnwidth, page=6]{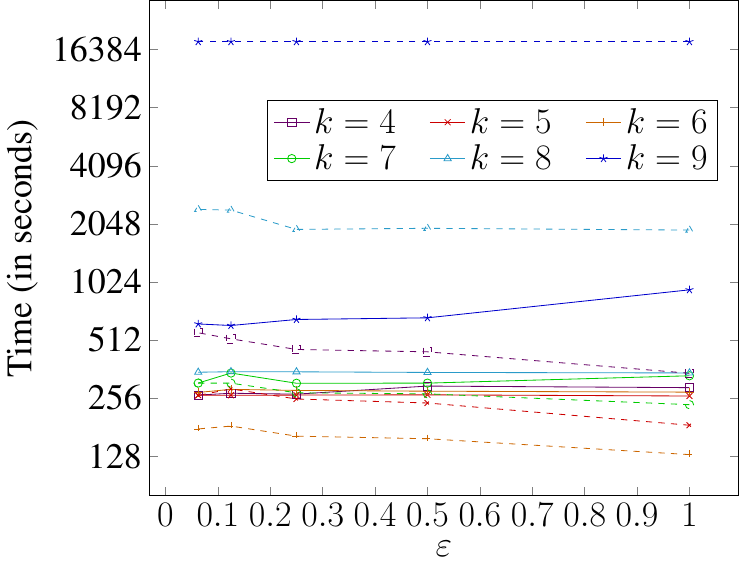}
  \caption{Parallel runtimes on com-friendster for approximate $k$-clique peeling using \ourapproxpeel (solid lines) and \kclist{} (dashed lines). These runtimes were obtained by varying over $\epsilon$, giving a $1/(k(1+\epsilon))$-approximation of the $k$-clique densest subgraph. These runtimes were obtained using the orientation given by degree ordering, and the runtimes are given in log-scale. Moreover, we cut off \kclist{}'s runtimes at 5 hours, which occurred for $k=9$ over all $\varepsilon$.}\label{fig:approx_peel_friend} 
\end{figure*}

\end{document}